\documentclass[10pt,twocolumn,twoside]{IEEEtran}
\usepackage{blindtext}
\usepackage{graphicx}
\usepackage{subfigure}
\usepackage{amsmath,amssymb}
\usepackage[noadjust]{cite} 
\usepackage{epstopdf}
\usepackage{algorithm}
\usepackage{algorithmic}
\usepackage{enumerate,mdwlist}
\usepackage{xcolor}
\newtheorem{corollary}{Corollary}
\newtheorem{assumption}{Assumption}
\newtheorem{remark}{Remark}
\newtheorem{definition}{Definition}
\newtheorem{lemma}{Lemma}
\newtheorem{theorem}{Theorem}
\newtheorem{example}{Example}

\newenvironment{proof}[1][Proof]{\noindent\textbf{#1.} }{\hfill \rule{0.5em}{0.5em}}
\begin{document}
%
\title{A Directed Spanning Tree Adaptive Control Framework for Time-Varying Formations
\thanks{\textcolor[rgb]{1.00,0.00,0.00}{The revised version of this work has been accepted by IEEE Transactions on Control of Network Systems, doi: 10.1109/TCNS.2021.3050332.}}
}

\author{Dongdong Yue,~\IEEEmembership{}
        Simone Baldi,~\IEEEmembership{Senior Member,~IEEE,}
        Jinde Cao,~\IEEEmembership{Fellow,~IEEE,} \\
        Qi Li,~\IEEEmembership{} and
        Bart De Schutter,~\IEEEmembership{Fellow,~IEEE}
\thanks{D. Yue and Q. Li are with School of Automation, and Key Laboratory of Measurement and Control of CSE, Ministry of Education, Southeast University, Nanjing, China (e-mail: yueseu@gmail.com; liq@jstd.gov.cn).}
\thanks{S. Baldi is with School of Mathematics, Southeast University, Nanjing, China and with Delft Center for Systems and Control, Delft University of
Technology, Delft, The Netherlands (e-mail: S.Baldi@tudelft.nl).}
\thanks{J. Cao is with School of Mathematics, and Jiangsu Provincial Key Laboratory of Networked Collective Intelligence, Southeast University, Nanjing, China (e-mail: jdcao@seu.edu.cn).}
\thanks{B. De Schutter is with Delft Center for Systems and Control, Delft University of Technology, Delft, The Netherlands (e-mail: B.DeSchutter@tudelft.nl).}
}
\maketitle

\begin{abstract}
In this paper, the time-varying formation and time-varying formation tracking problems are solved for linear multi-agent systems over digraphs without the knowledge of the eigenvalues of the Laplacian matrix associated to the digraph. The solution to these problems relies on a framework that generalizes the directed spanning tree adaptive method, which was originally limited to consensus problems. Necessary and sufficient conditions for the existence of solutions to the formation problems are derived. Asymptotic convergence of the formation errors is proved via graph theory and Lyapunov analysis.
\end{abstract}

\begin{IEEEkeywords}
Adaptive control, directed graphs, multi-agent systems, formation control.
\end{IEEEkeywords}

%
\IEEEpeerreviewmaketitle

\section{Introduction}
\label{intro}
Formation control of multi-agent systems has captured increasing attention due to applications in spacecraft formation flying, search and rescue operations, intelligent transport system, to name a few \cite{oh2015survey,hu2020cooperative}. By designing appropriate feasibility conditions, results on time-varying formation (TVF) \cite{liu2012iterative,brinon2014cooperative,dong2014formation}, and time-varying formation tracking (TVFT) \cite{dong2018time,dong2017time,yu2018practical} have extended the time-invariant formation case. These designs rely on consensus-based methodologies \cite{ren2006consensus,xiao2009finite,yu2011second,baldi2018output} to accomplish the formation in a \emph{distributed} way (i.e. using local information only). However, a common notable problem in such methods is the required knowledge of the smallest nonzero eigenvalue of the communication Laplacian matrix, which might be unknown in large networks.

It is known that by suitably designing time-varying coupling weights in the network, the knowledge of the Laplacian eigenvalues can be overcome: this was shown for consensus \cite{yu2012distributed,li2013distributed,cheng2018fully}, containment \cite{wen2017robust}, or TVF \cite{wang2017distributed,wang2018distributed,yue2020distributed} problems over undirected or detail-balanced/strongly-connected digraphs. For more general digraphs, the analysis is challenging due to the complexity of the Laplacian. To address this complexity, a distributed adaptive control method has recently been studied for synchronization/consensus problems in \cite{yu2015distributed,yu2018distributed,yu2018directed}: this method exploits the presence of a directed spanning tree (DST) in the network. However, a unifying DST-based adaptive control framework encompassing TVF and TVFT problems is not available. Most notably, it is unclear how to design appropriate feasibility conditions for time-varying formations in the DST framework. These observations motivate this study.

The main contribution of this paper is a unifying DST-based adaptive control framework addressing TVF and TVFT: not only does the proposed framework still avoid the knowledge of the Laplacian eigenvalues, but it also help to establish necessary and sufficient conditions for such time-varying formations from a different perspective. For TVF without leaders, a novel class of feasibility conditions is proposed, which is more efficient to check than the feasibility conditions in the state of the art. The proposed conditions generalize in a natural unified way in the presence of one or more leaders.

The paper is organized as follows: Section \ref{Preli} gives some preliminaries and formulates the problems. Sections \ref{tvfsec}-\ref{tvftsec} present the main results for TVF and TVFT, respectively. Numerical examples are provided in Section \ref{exsec}. Section \ref{consec} concludes this paper.

\section{Preliminaries and Problem Statement}
\label{Preli}
\subsection{Notations}
  Let $\mathbb{R}$, $\mathbb{R}^+$, $\mathbb{R}^{n}$, $\mathbb{R}^{n\times p}$ represent the sets of real scalars, real positive scalars, $n$-dimensional column vectors, $n\times p$ matrices, respectively. Let $\textbf{I}_{n}$ and $\textbf{1}_{n}$ be the $n\times n$ identity matrix, and the column vector with $n$ elements being one, respectively. Zero vectors and zero matrices are all denoted by $0$. For a vector $x$, let $\|x\|$ denote the Euclidean norm. For a real symmetric matrix $A$, $\lambda_{\text{M}}(A)$ (resp. $\lambda_{\text{m}}(A)$) is its maximum (resp. minimum) eigenvalue, and $A>0$ (resp. $A\geq0$) means that $A$ is positive definite (resp. semi-definite). Denote $\mathcal{I}_N=\{1,2,\cdots,N\}$ as the set of natural numbers up to $N$. Denote $\text{col}(x_1,\cdots,x_N)=({x_1}^T,\cdots,{x_N}^T)^T$ as the column vectorization. The abbreviation $\text{diag}(\cdot)$ is the diagonalization operator and 'N-S' is short for 'necessary and sufficient'. The cardinality of a set is denoted by $|\cdot|$ and the difference (resp. union) of the sets $\mathcal{S}_1$ and $\mathcal{S}_2$ is denoted by $\mathcal{S}_1\setminus\mathcal{S}_2$ (resp. $\mathcal{S}_1\bigcup\mathcal{S}_2$). Moreover, $\otimes$ stands for the Kronecker product.

\subsection{Graph Theory}
 A weighted digraph $\mathcal{G}(\mathcal{V},\mathcal{E},\mathcal{A})$ is specified by the node set $\mathcal{V}=\{1,\cdots,N\}$, the edge set $\mathcal{E}=\{e_{ij}|i\rightarrow j, i\neq j\}$ and the weighted adjacency matrix $\mathcal{A}=(a_{ij})\in\mathbb{R}^{N\times N}$. In the matrix $\mathcal{A}$, $a_{ij}>0$ if $e_{ji}\in\mathcal{E}$, indicating that $j$ (resp. $i$) is an in-neighbor (resp. out-neighbor) of $i$ (resp. $j$), which can be denoted by $j\in\mathcal{N}_{1}(i)$ (resp. $i\in\mathcal{N}_{2}(j)$). Let $\mathcal{D}_2(i)=|\mathcal{N}_{2}(i)|$ be the out-degree of $i$. Moreover, $\mathcal{L}=(\mathcal{L}_{ij})\in\mathbb{R}^{N\times N}$ is the Laplacian matrix of $\mathcal{G}$, which is defined as: $\mathcal{L}_{ij}=-a_{ij}$, if $i\neq j$, and $\mathcal{L}_{ii}=\sum_{k=1,k\neq i}^{N}a_{ik}$, $\forall i\in\mathcal{I}_N$. A path of $\mathcal{G}$ from node $1$ to $s$ corresponds to an ordered sequence of edges $(e_{1,p_1},e_{p_1,p_2},\cdots,e_{p_{s},s})$. A digraph $\mathcal{G}$ is \textit{weakly-connected} if every pair of nodes are connected by a path disregarding the directions. A \textit{directed spanning tree (DST)} of $\mathcal{G}$ is a subgraph where there is a node called the root, that has no in-neighbors, such that one can find a path from the root to every other node. In a DST, if $j$ is an in-neighbor of $i$, one can also say that $j$ is a parent node, and $i$ is a child node. Moreover, a node is called a stem if it has at least one child, and a leaf otherwise.

\subsection{Problem Statement}
 Let $\mathcal{G}(\mathcal{V},\mathcal{E},\mathcal{A})$ denote the digraph that characterizes the communication topology among $N$ agents, where the weights in $\mathcal{A}$ represent the communication strengths. The dynamics of the agents are given by
 \begin{align}\label{agents}
   \dot{x}_i=Ax_i+Bu_i,  \quad i\in\mathcal{I}_N
 \end{align}
 where $x_i\in\mathbb{R}^n$ is the state of agent $i$ and $u_i\in\mathbb{R}^m$ is its control input to be designed. Let the pair $(A,B)$ be stabilizable.

 \begin{definition}[TVF]\label{tvf}
    The multi-agent system (\ref{agents}) is said to achieve the time-varying formation (TVF) defined by the time-varying vector $h(t)=\text{col}(h_1(t),h_2(t),\cdots,h_N(t))$ if, for any initial states, there holds
    \begin{align}\label{tvfdef}
       \lim_{t\rightarrow\infty}((x_i-h_i)-(x_j-h_j))&=0, \ \forall i,j\in\mathcal{I}_N.
    \end{align}
 \end{definition}

 Now consider the case where there are $M$ leader agents, $M\geq 1$, in the network $\mathcal{G}$. Without loss of generality, let the first $M$ agents be the leaders, and the rest be the followers:
  \begin{align}\label{leaders}
   &\dot{x}_l=Ax_l,  \quad\quad\quad\quad l\in\mathcal{I}_M,        \nonumber\\
   &\dot{x}_i=Ax_i+Bu_i, \quad i\in\mathcal{I}_N\setminus\mathcal{I}_M.
  \end{align}
 As leaders have no in-neighbors, the Laplacian matrix of $\mathcal{G}$ can be partitioned as
 \begin{equation}\label{lhat}       
  \mathcal{L}=\left(                 
  \begin{array}{cc}   
      0 & 0 \\
    \mathcal{L}_1 & \mathcal{L}_2   
  \end{array}
 \right)                 
 \end{equation}
 where $\mathcal{L}_1\in\mathbb{R}^{(N-M)\times M}$ and $\mathcal{L}_2\in\mathbb{R}^{(N-M)\times(N-M)}$.

\begin{definition}[\cite{dong2017time}]\label{wellinf}
  A follower is called \emph{well-informed} if all leaders are its in-neighbors, and is \emph{uninformed} if no leader is its in-neighbor.
\end{definition}

 \begin{definition}[TVFT]\label{tvft}
    The multi-agent system (\ref{leaders}) is said to achieve the time-varying formation tracking (TVFT) defined by the time-varying vector $h^F(t)=$ $\text{col}(h_{M+1}(t),h_{M+2}(t),\cdots,h_N(t))$ and by positive constants $\beta_l$, $l\in\mathcal{I}_M$, satisfying $\sum_{l=1}^{M}\beta_l=1$ if, for any initial states, there holds
    \begin{align}\label{tvftml}
      \lim_{t\rightarrow\infty}\big(x_i-h_i-\sum_{l=1}^M\beta_lx_{l}\big)=0, \ \forall i\in\mathcal{I}_N\setminus\mathcal{I}_M.
    \end{align}
    For the special case $M=1$, (\ref{tvftml}) becomes
    \begin{align}\label{tvftsl}
      \lim_{t\rightarrow\infty}\big(x_i-h_i-x_1\big)=0, \quad i=2,\cdots,N.
    \end{align}
 \end{definition}

  The goal of this paper is to solve the problems outlined by (\ref{tvfdef}), (\ref{tvftml}) and (\ref{tvftsl}) without the knowledge of the Laplacian eigenvalues, by consistently generalizing the DST idea.

\section{DST-Based Distributed Adaptive TVF}\label{tvfsec}
 This section appropriately extends the DST-based adaptive control method to solve the TVF problem of Definition \ref{tvf}. The following is a standard connectivity assumption (\cite{liu2012iterative,dong2014formation}, etc).
  \begin{assumption}\label{dst}
   The digraph $\mathcal{G}$ has at least one DST.
 \end{assumption}

 Under Assumption \ref{dst}, one can select a DST $\bar{\mathcal{G}}(\mathcal{V},\bar{\mathcal{E}},\bar{\mathcal{A}})$ of $\mathcal{G}$. Note that finding a DST can also be conducted in a distributed manner, but it requires the agents to exchange more information. As in \cite{yu2018distributed}, we assume that $\bar{\mathcal{G}}$ is known. Without loss of generality, let node $1$ be the root of the $\bar{\mathcal{G}}$. Correspondingly, let $\bar{\mathcal{L}}$ be the Laplacian matrix of $\bar{\mathcal{G}}$ and $\bar{\mathcal{N}}_{\text{2}}(i)$ be the set of out-neighbors of $i$ in $\bar{\mathcal{G}}$.

 Let $i_k$ denote the unique parent of node $k+1$ in $\bar{\mathcal{G}}$ for $k\in\mathcal{I}_{N-1}$, then $\bar{\mathcal{E}}=\{e_{i_k,k+1}|k\in\mathcal{I}_{N-1}\}\subset\mathcal{E}$. For compactness, define $d_i(t)=x_i(t)-h_i(t)$ as the formation state, i.e., the distance between the current state and the desired formation offset of agent $i$. Denote $x=\text{col}(x_1,\cdots,x_N)$, $d=\text{col}(d_1,\cdots,d_N)$.

 We propose the DST-based adaptive TVF controller as:
 \begin{align}
   \label{tvfui}
   &u_i=K_0x_i+K_1d_i+K_2\sum_{j\in\mathcal{N}_{1}(i)}\alpha_{ij}(t)(d_i-d_j)
 \end{align}
 with the time-varying coupling weights
 \begin{align}
   \label{tvfapij}
   &\alpha_{ij}(t)=\left\{
                    \begin{array}{ll}
                      a_{ij}, & \text{if}\quad e_{ji}\in\mathcal{E}\setminus\bar{\mathcal{E}}, \\
                      \bar{a}_{k+1,i_k}(t), & \text{if}\quad  e_{ji}\in\bar{\mathcal{E}}.
                    \end{array}
                  \right. \\
   \label{alphabardot}
   &\dot{\bar{a}}_{k+1,i_k}=\rho_{k+1,i_k}\Big((d_{i_k}-d_{k+1})-   \nonumber\\
   & \quad\qquad\qquad \sum\limits_{j\in\bar{\mathcal{N}}_{2}(k+1)}(d_{k+1}-d_j)\Big)^T\Gamma(d_{i_k}-d_{k+1}).
 \end{align}
 In (\ref{tvfui})-(\ref{alphabardot}), $K_0$, $K_1$, $K_2$, and $\Gamma$ are gains to be designed, and $\rho_{k+1,i_k}\in\mathbb{R}^+$. In (\ref{tvfui}), $\alpha_{ij}(t)$ is the coupling weight between agent $i$ and its in-neighbor $j$, which is time-varying only if the corresponding edge appears in $\bar{\mathcal{G}}$, i.e., $j=i_k$ and $i=k+1$ for some $k\in\mathcal{I}_{N-1}$, and constant otherwise.

\begin{remark}\label{k0k1k2}
   The structure of controller (\ref{tvfui}) is as follows. The gain $K_0$ is to be designed to make the time-varying formation $h(\cdot)$ feasible; the gain $K_1$ is needed to control the average formation signal $d_{\text{ave}}=\frac{1}{N}\sum_{j\in\mathcal{I}_N}d_j$; the gain $K_2$ is a consensus gain. Different from the related literature \cite{liu2012iterative,dong2014formation}, the DST structure is explicitly used in the control law (\ref{tvfui})-(\ref{alphabardot}).
\end{remark}
 \subsection{Technical lemmas}
 \begin{lemma}[N-S condition for TVF]\label{lmaXi}
  Under Assumption \ref{dst}, and for any DST $\bar{\mathcal{G}}$, define $\Xi\in\mathbb{R}^{(N-1)\times N}$ as
  \begin{align}\label{Xi}
    \Xi_{kj}=\left\{
              \begin{array}{ll}
                -1, & \text{if}\quad j=k+1,  \\
                1, & \text{if}\quad j=i_k,  \\ 
                0, & \text{otherwise}.
              \end{array}
            \right.
 \end{align}
 Then, the TVF for multi-agent system (\ref{agents}) can be achieved if and only if
   \begin{align}\label{tvfns}
    \lim_{t\rightarrow\infty}\|(\Xi\otimes\textbf{I}_n)d(t)\|=0.
    \end{align}
\end{lemma}
 \begin{proof}
  From Lemma 3.2 in \cite{yu2018distributed}, (\ref{tvfns}) holds if and only if $\lim_{t\rightarrow\infty}\|d_i(t)-d_j(t)\|=0, \forall i,j\in\mathcal{I}_N$. Then, Lemma \ref{lmaXi} holds following Definition \ref{tvf} and the definition of $d_i(t)$. In fact, $\Xi^T$ is the incidence matrix associated to $\bar{\mathcal{G}}$.
\end{proof}

 \begin{lemma}[Auxiliary matrix $Q$]\label{lmaQ}
  Under Assumption \ref{dst}, and for any DST $\bar{\mathcal{G}}$, define $Q\in\mathbb{R}^{(N-1)\times(N-1)}$ as $Q=\tilde{Q}+\bar{Q}$ with $\tilde{Q}_{kj}=\sum_{c\in\bar{\mathcal{V}}_{j+1}}(\tilde{\mathcal{L}}_{k+1,c}-\tilde{\mathcal{L}}_{i_k,c})$ and $\bar{Q}_{kj}=\sum_{c\in\bar{\mathcal{V}}_{j+1}}(\bar{\mathcal{L}}_{k+1,c}-\bar{\mathcal{L}}_{i_k,c}).$ Here, $\bar{\mathcal{V}}_{j+1}$ represents the vertex set of the subtree rooting at node $j+1$ and $\tilde{\mathcal{L}}=\mathcal{L}-\bar{\mathcal{L}}$. Then, there holds
  \begin{align}\label{xlqx}
    \Xi\mathcal{L}=Q\Xi
  \end{align} where $\Xi$ is defined in (\ref{Xi}). Moreover, $\bar{Q}$ can be explicitly written as
  \begin{align}\label{qbar}
      \bar{Q}_{kj}=\left\{
                   \begin{array}{ll}
                     \bar{a}_{j+1,i_j}, & \text{if}\quad  j=k,  \\
                     -\bar{a}_{j+1,i_j}, & \text{if}\quad  j=i_k-1,  \\ 
                     0, & \text{otherwise}.
                   \end{array}
                 \right.
  \end{align}
\end{lemma}

\begin{proof}
    See the appendix. The proof revises and completes the results in \cite{yu2018distributed}, \cite{yu2018directed}, since step 1) of the proof ($\mathcal{L}=\mathcal{L}J\Xi$) is missing there.
\end{proof}

\begin{remark}\label{rmq}
  Lemma \ref{lmaQ} states that the information of the Laplacian $\mathcal{L}$ can be transferred into a reduced-order matrix $Q$ through a commutative-like multiplication law (\ref{xlqx}). For the off-diagonal elements of $\bar{Q}$, $\bar{Q}_{kj}=-\bar{Q}_{jj}$ if and only if $j+1$ is the parent of $k+1$ in $\bar{\mathcal{G}}$.
\end{remark}

 \begin{lemma}[Feasibility conditions]\label{lmatvf}%
   Under Assumption \ref{dst}, let us consider controller (\ref{tvfui}) with time-varying coupling weights (\ref{tvfapij}) for any DST $\bar{\mathcal{G}}$. Suppose that the origin of the linear time-varying system
 \begin{align}\label{dbdot}
  \dot{d}_L=(\textbf{I}_{N-1}\otimes(A+BK_0+BK_1)+Q(t)\otimes BK_2)d_L    \nonumber\\
\end{align}
   is globally asymptotically stable, where $Q(t)=\tilde{Q}+\bar{Q}(t)$ with fixed $\tilde{Q}$ defined as in Lemma \ref{lmaQ}, and
\begin{align}\label{qb}
 \bar{Q}_{kj}(t)=\left\{
                   \begin{array}{ll}
                     \bar{a}_{j+1,i_j}(t), &\text{if}\quad  j=k,  \\
                     -\bar{a}_{j+1,i_j}(t), &\text{if}\quad  j=i_k-1,  \\ 
                     0, & \text{otherwise}.
                   \end{array}
                 \right.
\end{align}
 Then, the TVF problem can be solved by controller (\ref{tvfui}) if and only if
    \begin{align}\label{fes}
       \lim_{t\rightarrow\infty}(A+BK_0)(h_{i_k}&(t)-h_{k+1}(t)) \nonumber\\
                                       &-(\dot{h}_{i_k}(t)-\dot{h}_{k+1}(t))=0
     \end{align}
     holds $\forall k\in\mathcal{I}_{N-1}$.
\end{lemma}

\begin{proof}
    Let $\bar{d}_k(t)=d_{i_k}(t)-d_{k+1}(t)$ be the error vector between the parent and the child nodes of the directed edge $e_{i_k,k+1}$, $k\in\mathcal{I}_{N-1}$, and denote $\bar{d}=\text{col}(\bar{d}_1,\cdots,\bar{d}_N)$. Then, $\bar{d}=(\Xi\otimes\textbf{I}_n)d$. From Lemma \ref{lmaXi}, it remains to prove that $\lim_{t\rightarrow\infty}\|\bar{d}(t)\|=0$ under the given conditions.

 Based on (\ref{agents}) and (\ref{tvfui}), the dynamics of $x(t)$ is given by
 \begin{align}\label{xdot}
  \dot{x}=(\textbf{I}_N\otimes(A+BK_0&+BK_1))x+(\mathcal{L}(t)\otimes BK_2)d    \nonumber\\
                                        &-(\textbf{I}_N\otimes BK_1)h
\end{align}
where $\mathcal{L}(t)$ is the Laplacian matrix of $\mathcal{G}$ at time $t$ due to the adaptive mechanisms. Then, it follows from (\ref{xdot}) and the definitions of $d$ and $\bar{d}$ that
 \begin{align}\label{dtdot}
  \dot{\bar{d}}=&(\textbf{I}_{N-1}\otimes(A+BK_0+BK_1))\bar{d}+(\Xi\mathcal{L}(t)\otimes BK_2)d    \nonumber\\
                  &+(\Xi\otimes(A+BK_0))h-(\Xi\otimes \textbf{I}_n)\dot{h}      \nonumber\\
                =&(\textbf{I}_{N-1}\otimes(A+BK_0+BK_1)+Q(t)\otimes BK_2)\bar{d}    \nonumber\\
                  &+(\Xi\otimes(A+BK_0))h-(\Xi\otimes \textbf{I}_n)\dot{h}
\end{align}
 where Lemma \ref{lmaQ} is used to get the second equality. Given that the linear system (\ref{dbdot}) asymptotically converges to zero, one knows that $\lim_{t\rightarrow\infty}\|\bar{d}(t)\|=0$ if and only if
\begin{align}\label{fexi}
    \lim_{t\rightarrow\infty}(\Xi\otimes(A+BK_0))h(t)-(\Xi\otimes \textbf{I}_n)\dot{h}(t)=0.
\end{align}

From the definition of $\Xi$, condition (\ref{fes}) is equivalent to (\ref{fexi}). This completes the proof.
\end{proof}

\subsection{Main result}
The design process of the TVF controller is summarized in Algorithm \ref{tvfalg}, and analyzed in the following theorem.
 \begin{algorithm}\caption{TVF Controller Design}\label{tvfalg}
  \begin{enumerate}
     \item Find a constant $K_0$ such that the formation feasibility condition
    \begin{align}\label{fescon}
       (A+BK_0)(h_{i_k}&(t)-h_{k+1}(t)) \nonumber\\
                                       &-(\dot{h}_{i_k}(t)-\dot{h}_{k+1}(t))=0
     \end{align}
     holds $\forall k\in\mathcal{I}_{N-1}$ for any DST $\bar{\mathcal{G}}$. If such $K_0$ exists, continue; else, the algorithm terminates without solutions;
     \item\label{s2a1} Choose $K_1$ such that $(A+BK_0+BK_1,B)$ is stabilizable (using, e.g., pole placement). For some $\eta$, $\theta\in\mathbb{R}^+$, solve the following LMI:
     \begin{align}\label{mainlmi}       
        (A+BK_0+BK_1)P+&P(A+BK_0+BK_1)^T   \nonumber\\
                       &-\eta BB^T+\theta P\leq0
     \end{align}
     to get a $P>0$;
     \item\label{s3a1} Set $K_2=-B^TP^{-1}$, $\Gamma=P^{-1}BB^TP^{-1}$ and choose scalars $\rho_{k+1,i_k}\in\mathbb{R}^+$.
  \end{enumerate}
 \end{algorithm}

 \begin{theorem}[Main result for TVF]\label{tvfthrem}
   Under Assumption \ref{dst}, and feasibility condition (\ref{fescon}), the TVF problem in Definition \ref{tvf} is solved by controller (\ref{tvfui}) with adaptive coupling weights (\ref{tvfapij})-(\ref{alphabardot}), along the designs in Algorithm \ref{tvfalg}.
 \end{theorem}

 \begin{proof}
 The feasibility condition (\ref{fescon}) guarantees that (\ref{fes}) holds $\forall k\in\mathcal{I}_{N-1}$. Moreover,
 \begin{align}\label{dtdot2}
  \dot{\bar{d}}=&(\textbf{I}_{N-1}\otimes(A+BK_0+BK_1)+Q(t)\otimes BK_2)\bar{d},
\end{align}
where $Q(t)$ is defined as in Lemma \ref{lmatvf} based on $\bar{\mathcal{G}}$. In the following, it will be proved that the designed controller guarantees $\lim_{t\rightarrow\infty}\bar{d}(t)=0$. As such, the proof of the theorem will be complete according to Lemma \ref{lmatvf}.

Consider the Lyapunov candidate
\begin{align}\label{V1}
  V_1(t)=\frac{1}{2}\bar{d}^T(&\textbf{I}_{N-1}\otimes P^{-1})\bar{d}  \nonumber\\
                       & +\sum_{k=1}^{N-1}\frac{1}{2\rho_{k+1,i_k}}(\bar{a}_{k+1,i_k}(t)-\phi_{k+1,i_k})^2
\end{align}
where $P$ is a solution to (\ref{mainlmi}) and $\phi_{k+1,i_k}\in\mathbb{R}^{+}$, $k\in\mathcal{I}_{N-1}$ are to be decided later.

By (\ref{dtdot2}) and (\ref{alphabardot}), the derivative of $V_1$ is
\begin{align}\label{v1d1}
  \dot{V}_1 =&\bar{d}^T(\textbf{I}_{N-1}\otimes P^{-1}(A+BK_0+BK_1)  \nonumber\\
  &\qquad\qquad\qquad\qquad+Q(t)\otimes P^{-1}BK_2)\bar{d}      \nonumber\\
  +&\sum_{k=1}^{N-1}(\bar{a}_{k+1,i_k}-\phi_{k+1,i_k})(\bar{d}_k-\sum_{j+1\in\bar{\mathcal{N}}_{2}(k+1)}\bar{d}_j)^T\Gamma\bar{d}_k.
\end{align}

Based on Lemma \ref{lmaQ}, one has
\begin{align}\label{abqb}
  &\sum_{k=1}^{N-1}\bar{a}_{k+1,i_k}(\bar{d}_k-\sum_{j+1\in\bar{\mathcal{N}}_{2}(k+1)}\bar{d}_j)^T\Gamma\bar{d}_k    \nonumber\\
 =&\sum_{k=1}^{N-1}(\bar{Q}_{kk}(t)\bar{d}_k+\sum_{j=1,j\neq k}^{N-1}\bar{Q}_{jk}(t)\bar{d}_j)^T\Gamma\bar{d}_k      \nonumber\\
 =&\sum_{k=1}^{N-1}\sum_{j=1}^{N-1}\bar{Q}_{jk}(t)\bar{d}_j^T\Gamma\bar{d}_k
\end{align}

Let us define $\Phi\in\mathbb{R}^{(N-1)\times(N-1)}$ as
\begin{align}\label{Phi}
 \Phi_{kj}=\left\{
                   \begin{array}{ll}
                     \phi_{j+1,i_j}, &\text{if}\quad j=k,  \\
                     -\phi_{j+1,i_j}, &\text{if}\quad j=i_k-1,  \\ 
                     0, & \text{otherwise}.
                   \end{array}
                 \right.
\end{align}
Then, it follows from (\ref{v1d1})-(\ref{Phi}) that
 \begin{align}\label{v1d2}
  \dot{V}_1 =&\bar{d}^T(\textbf{I}_{N-1}\otimes P^{-1}(A+BK_0+BK_1)  \nonumber\\
  &\qquad\qquad\qquad\qquad+Q(t)\otimes P^{-1}BK_2)\bar{d}      \nonumber\\
  &+\sum_{k=1}^{N-1}\sum_{j=1}^{N-1}(\bar{Q}_{jk}(t)-\Phi_{jk})\bar{d}_j^T\Gamma\bar{d}_k    \nonumber\\
  =&\bar{d}^T(\textbf{I}_{N-1}\otimes P^{-1}(A+BK_0+BK_1)  \nonumber\\
  &\qquad\qquad\qquad\qquad+Q(t)\otimes P^{-1}BK_2)\bar{d}      \nonumber\\
  &+\bar{d}^T((\bar{Q}(t)-\Phi)\otimes\Gamma)\bar{d}.
\end{align}

Define $\tilde{d}=(\textbf{I}_{N-1}\otimes P^{-1})\bar{d}$, and substitute $K_2,\Gamma$ designed in Algorithm \ref{tvfalg} into (\ref{v1d2}). Then, one has
 \begin{align}\label{v1d3}
  \dot{V}_1=&\tilde{d}^T(\textbf{I}_{N-1}\otimes (A+BK_0+BK_1)P  \nonumber\\
  &\qquad\qquad\qquad\qquad-Q(t)\otimes BB^T)\tilde{d}      \nonumber\\
  &+\tilde{d}^T((\bar{Q}(t)-\Phi)\otimes BB^T)\tilde{d}       \nonumber\\
  =&\tilde{d}^T(\textbf{I}_{N-1}\otimes (A+BK_0+BK_1)P)\tilde{d}  \nonumber\\
  &-\tilde{d}^T((\tilde{Q}+\Phi)\otimes BB^T)\tilde{d}     \nonumber\\
  =&\frac{1}{2}\tilde{d}^T\Big(\textbf{I}_{N-1}\otimes\big((A+BK_0+BK_1)P     \nonumber\\
  &\qquad\qquad\qquad+P(A+BK_0+BK_1)^T\big)  \nonumber\\
  &-(\tilde{Q}+\tilde{Q}^T+\Phi+\Phi^T)\otimes BB^T\Big)\tilde{d}.
\end{align}
Now we show that by appropriately selecting $\phi_{k+1,i_k}$, $k\in\mathcal{I}_{N-1}$, it can be fulfilled that
\begin{align}\label{Phis}
  &\Phi+\Phi^T=  \nonumber\\
            &\left(
                \begin{array}{ccccc}
                  2\phi_{2,i_1} & \phi_{21} & \cdots & \phi_{N-2,1} & \phi_{N-1,1} \\
                  \phi_{21} & 2\phi_{3,i_2} & \cdots & \cdots & \phi_{N-1,2} \\
                  \vdots & \vdots & \ddots & \vdots & \vdots \\
                  \phi_{N-2,1} & \vdots & \cdots & 2\phi_{N-1,i_{N-2}} & \phi_{N-1,N-2} \\
                  \phi_{N-1,1} & \phi_{N-1,2} & \cdots & \phi_{N-1,N-2} & 2\phi_{N,i_{N-1}} \\
                \end{array}
              \right)
\end{align}
is positive definite. To see this, let us denote $\Psi_1=\left(
                                                           \begin{array}{c}
                                                             2\phi_{2,i_1} \\
                                                           \end{array}
                                                         \right)
$ and $\Psi_k=\left(
                                                                            \begin{array}{cc}
                                                                              \Psi_{k-1} & \varphi_k \\
                                                                              \varphi_k^T & 2\phi_{k+1,i_k} \\
                                                                            \end{array}
                                                                          \right)$,
where $\varphi_k=(\phi_{k1},\phi_{k2},\cdots,\phi_{k,k-1})^T$, $k=2,\cdots,N-1$. Clearly, $\Psi_1>0$ by choosing $\phi_{2,i_1}>0$. Now suppose $\Psi_{k-1}>0$, $k\geq2$. Note that $|\phi_{kj}|\leq|\phi_{j+1,i_j}|$, $\forall j\in\mathcal{I}_{k-1}$. Then, one has $\varphi_k^T\Psi_{k-1}^{-1}\varphi_k\leq\lambda_{\text{M}}(\Psi_{k-1}^{-1})\sum_{j=2}^k\phi_{j,i_{j-1}}^2$. By choosing $\phi_{k+1,i_k}>\frac{\sum_{j=2}^k\phi_{j,i_{j-1}}^2}{2\lambda_{\text{m}}(\Psi_{k-1})}$, one has $\Psi_k>0$ according to the Schur complement \cite[Chapter 2.1]{boyd1994linear}. By mathematical induction, $\Phi+\Phi^T=\Psi_{N-1}$ is positive definite.

Moreover, since $\tilde{Q}$ is fixed, one can always choose sufficiently large $\phi_{k+1,i_k}$, $k\in\mathcal{I}_{N-1}$, such that $\lambda_{\text{m}}(\tilde{Q}+\tilde{Q}^T+\Phi+\Phi^T)\geq\eta$ where $\eta$ is defined in (\ref{mainlmi}). Then, it follows from (\ref{v1d3}) and (\ref{mainlmi}) that
 \begin{align}\label{v1d4}
  \dot{V}_1\leq&\frac{1}{2}\tilde{d}^T\Big(\textbf{I}_{N-1}\otimes\big((A+BK_0+BK_1)P     \nonumber\\
  &\qquad+P(A+BK_0+BK_1)^T-\eta BB^T\big)\Big)\tilde{d}   \nonumber\\
  \leq&-\frac{\theta}{2}\tilde{d}^T(\textbf{I}_{N-1}\otimes P)\tilde{d}
  =-\frac{\theta}{2}\bar{d}^T(\textbf{I}_{N-1}\otimes P^{-1})\bar{d}
  \leq0
\end{align}
which implies that the signals $\bar{d}(t)$ and $\bar{a}_{k+1,i_k}(t)$ in $V_1(t)$ are bounded. Note that $\dot{V}_1(t)=0$ implies that $\bar{d}=0$, thus by LaSalle's invariance principle, one has $\lim_{t\rightarrow\infty}\bar{d}(t)=0$. This completes the proof. 
 \end{proof}

  \begin{remark}\label{lmifes}
   The LMI (\ref{mainlmi}) is feasible for some $P>0$ if and only if $(A+BK_0+BK_1,B)$ is stabilizable, which can be realized since $(A,B)$ is stabilizable. Note that different formation vectors $h(\cdot)$ might lead to different solutions $P,K_0,K_1$.
  \end{remark}

 \begin{remark}\label{dstvsnode}
   In state-of-the-art TVF, the number of feasibility conditions is of the order $\frac{N(N-1)}{2}$ (i.e., one condition for each pair of connected agents) \cite{dong2014formation,wang2017distributed}. The proposed number of feasibility conditions in (\ref{fescon}) is $N-1$, i.e., exploiting the DST structure leads to the minimum number of conditions: note that $N-1$ is the minimum number of edges such that $\mathcal{G}$ is weakly-connected.
 \end{remark}

\section{DST-Based Distributed Adaptive TVFT}\label{tvftsec}
In this section, we propose a novel generalized DST-based adaptive controller to solve the TVFT problem of Definition \ref{tvft}. We address the general case with multiple leaders, and give a corollary for the special case with a single leader.
\begin{definition}\label{gdst}
   The digraph $\mathcal{G}$ is said to have a generalized DST rooting at the leadership, if the followers are either well-informed or uninformed, and for each uninformed follower, there exists at least one well-informed follower that has a directed path to it.
\end{definition}

\begin{assumption}\label{dstml}
  The digraph $\mathcal{G}$ has at least one generalized DST rooting at the leadership.
\end{assumption}

\begin{remark}
  TVFT with multiple leaders is also considered in \cite{dong2017time,hu2020distributed}, where it is required that the coupling weights from any leader to different well-informed followers are identical and known a priori. Assumption \ref{dstml} relaxes that requirement.
\end{remark}

\subsection{Auxiliary system, technical lemma and control law}
Let us introduce an auxiliary multi-agent system with an induced communication graph $\mathcal{G}'(\mathcal{V}',\mathcal{E}',\mathcal{A}')$. Define $\mathcal{V}'=\mathcal{I}_{N-M+1}$ where the agent with index $1$ is the leader and $\mathcal{E}'=\{e'_{1j},j>1|j+M-1$ is well-informed in $\mathcal{G}\}\bigcup\{e'_{jp},j,p>1|e_{j+M-1,p+M-1}\in\mathcal{E}\}$. The adjacency matrix $\mathcal{A}'=(a'_{jp})$ where $a'_{jp}>0$ if $e'_{pj}\in\mathcal{E}'$, and $a'_{jp}=0$ otherwise.

To clarify Assumption \ref{dstml} and the induced graph $\mathcal{G}'$, see Fig.~\ref{topex3}. It is clear that the multiple leaders are merged as a single joint leader in $\mathcal{G}'$.
 \begin{figure}[thb!]
  \centering
  \includegraphics[width=0.50\textwidth]{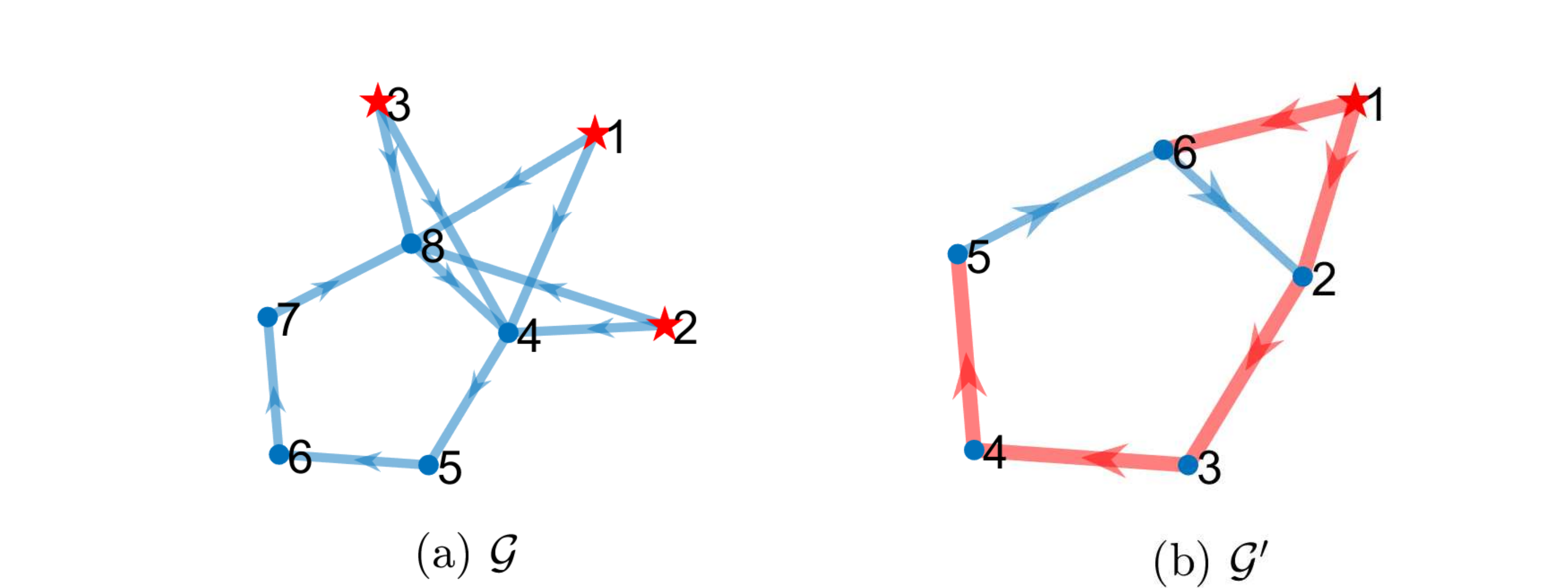}\\
  \caption{A communication graph $\mathcal{G}$ with three leaders (with indexes $1,2,3$) which satisfies Assumption \ref{dstml}, and the induced graph $\mathcal{G}'$ with a single leader (with index 1).}\label{topex3}
 \end{figure}

In the auxiliary multi-agent system, let $y_j$ and $v_j$ be the state and control input of agent $j$. For the leader, define $y_1=\sum_{l=1}^M\beta_lx_{l}$ and $h'_1\equiv0$. For the followers, define $y_j=x_{j+M-1}$, $h'_j=h_{j+M-1}$, for $j=2,\cdots,N-M+1$. Let $d'_j=y_j-h'_j$, $j\in\mathcal{I}_{N-M+1}$. Then, the dynamics of $y_j$ satisfies  %
\begin{align}\label{ydot}
  & \dot{y}_1=Ay_1,   \nonumber\\
  & \dot{y}_j=Ay_j+Bv_j  \quad j=2,\cdots,N-M+1,
\end{align}
where $v_j=u_{j+M-1}$, and the initial state values are determined by those of multi-agent system (\ref{leaders}).

\begin{lemma}[N-S condition for TVFT]\label{tvftmllma}
 Under Assumption \ref{dstml}, the multi-agent system (\ref{leaders}) achieves the TVFT with multiple leaders defined by $h^F(t)=\text{col}(h_{M+1}(t),$ $h_{M+2}(t),\cdots,h_N(t))$ and by $\beta_l$, $l\in\mathcal{I}_M$, if and only if the auxiliary system (\ref{ydot}) achieves the TVFT defined by $h'_F(t)=\text{col}(h'_2(t),h'_3(t),\cdots,h'_{N-M+1}(t))$ with a single leader.
\end{lemma}
\begin{proof}
  According to the definitions of $y_j$ and $h'_j$, it is obvious that $\lim_{t\rightarrow\infty}(y_j(t)-h'_j(t)-y_1(t))=0$, $j=2,\cdots,N-M+1$, is equivalent to $\lim_{t\rightarrow\infty}(x_i(t)-h_i(t)-\sum_{l=1}^M\beta_lx_{l}(t))=0$, $\forall i\in\mathcal{I}_N\setminus\mathcal{I}_M$.
\end{proof}

Under Assumption \ref{dstml}, there is at least one DST in $\mathcal{G}'$ rooting at the leader. Then, one can choose such a DST $\hat{\mathcal{G}}'(\mathcal{V}',\hat{\mathcal{E}}',\hat{\mathcal{A}}')$. Let $j_k$ denote the unique parent of node $k+1$ in $\hat{\mathcal{G}}'$ for $k\in\mathcal{I}_{N-M}$. Let $\mathcal{N}'_{1}(j)$ be the set of in-neighbors of $j$ in $\mathcal{G}'$ and $\hat{\mathcal{N}}'_{\text{2}}(j)$ be the set of out-neighbors of $j$ in $\hat{\mathcal{G}}'$.

The generalized DST-based distributed adaptive TVFT controller for follower $i$ of (\ref{leaders}), $i\in\mathcal{I}_N\setminus\mathcal{I}_M$, is proposed as:
 \begin{align}
   \label{tvftmlui}
   &u_i=v_{i-M+1},   \\
   \label{tvftmlapij}
   &v_j=K_0h'_j+K_2\sum_{p\in\mathcal{N}'_{1}(j)}\alpha'_{jp}(t)(d'_j-d'_p),  \\
   &\alpha'_{jp}(t)=\left\{
                    \begin{array}{ll}
                      a'_{jp}, & \text{if}\quad e_{pj}\in\mathcal{E}'\setminus\hat{\mathcal{E}}', \\
                      \hat{a}'_{k+1,j_k}(t), & \text{if}\quad  e_{pj}\in\hat{\mathcal{E}}'
                    \end{array}
                  \right.
 \end{align}
 \begin{align}\label{alphahatpdot}
   &\dot{\hat{a}}'_{k+1,j_k}=\rho_{k+1,j_k}\Big((d'_{j_k}-d'_{k+1})-\nonumber\\
   &\quad\qquad\qquad \sum\limits_{p\in\hat{\mathcal{N}}'_{2}(k+1)}(d'_{k+1}-d'_p)\Big)^T\Gamma(d'_{j_k}-d'_{k+1}).
 \end{align}

 In order to illustrate the idea of the auxiliary multi-agent system, the information flow of the closed-loop system $x_i$, $i\in\mathcal{I}_N\setminus\mathcal{I}_M$, is sketched in Fig.~\ref{cbd}. Instead of directly designing the controllers for multi-agent system (\ref{leaders}), an auxiliary multi-agent system is defined as in (\ref{ydot}), and some interaction between them is constructed: at stage (\ref{tvftmlapij}), each leader $x_l$ of (\ref{leaders}) broadcast its $\beta_l$-scaled state to the single leader of (\ref{ydot}), and each follower broadcast its state to the corresponding follower, respectively; at stage (\ref{tvftmlui}), each follower of (\ref{ydot}) responds to the corresponding follower of (\ref{leaders}) with its control input. Then, the original TVFT problem in (\ref{leaders}) is successfully transformed into the TVFT with a single leader in (\ref{ydot}). It should be pointed out that only the local information, i.e., the states of $x_s$, $s\in\mathcal{N}_1(i)$, are included in the loop of $x_i$ from Fig.~\ref{cbd}.
 \begin{figure}[thb!]
  \centering
  \includegraphics[width=0.32\textwidth]{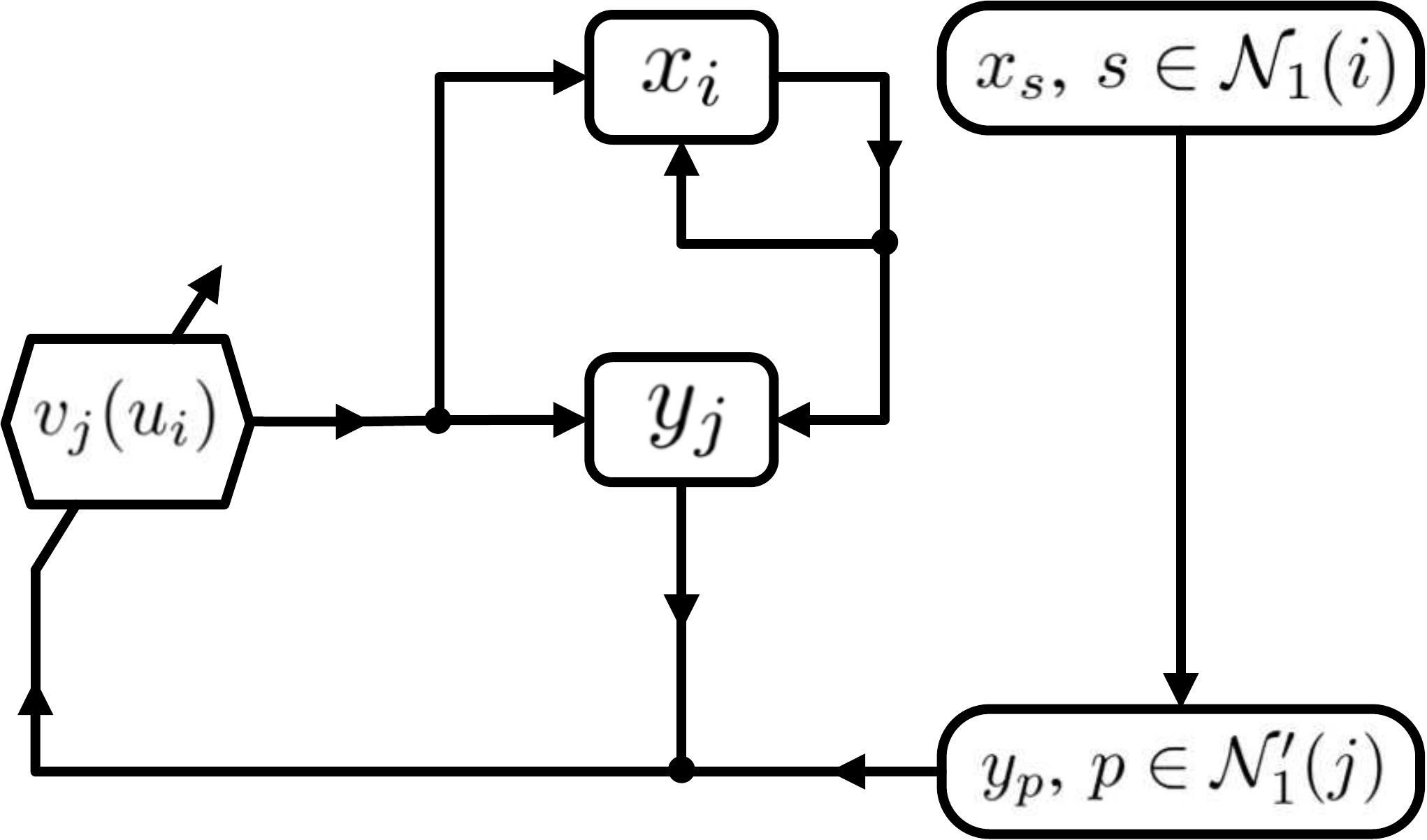}\\
  \caption{The information flow of the closed-loop system $x_i$, $i\in\mathcal{I}_N\setminus\mathcal{I}_M$.}\label{cbd}
 \end{figure}

\subsection{Main result}
The design process of the TVFT controller is summarized in Algorithm \ref{tvftalg}, and analyzed in the following theorem.
\begin{algorithm}\caption{TVFT Controller Design}\label{tvftalg}
  \begin{enumerate}
     \item Find a constant $K_0$ such that the formation tracking feasibility condition
     \begin{align}\label{fescon2}
       (A+BK_0)h_i(t)-\dot{h}_i(t)=0
     \end{align}
     holds $\forall i\in\mathcal{I}_N\setminus\mathcal{I}_M$. If such $K_0$ exists, continue; else, the algorithm terminates without solutions;
     \item\label{s2a3} Choose $\eta$, $\theta\in\mathbb{R}^+$, and solve the following LMI:
     \begin{align}\label{mainlmi2}       
        AP+PA^T-\eta BB^T+\theta P\leq0
     \end{align}
     to get a $P>0$;
    \item\label{s3a3} Set $K_2=-B^TP^{-1}$, $\Gamma=P^{-1}BB^TP^{-1}$ and choose scalars $\rho_{k+1,i_k}\in\mathbb{R}^+$.
  \end{enumerate}
 \end{algorithm}

 \begin{theorem}[Main result for TVFT]\label{tvftmlthrem}
   Under Assumption \ref{dstml}, and feasibility condition (\ref{fescon2}). The TVFT problem in Definition \ref{tvft} can be solved by controller (\ref{tvftmlui})-(\ref{alphahatpdot}) with $\rho_{k+1,j_k}\in\mathbb{R}^+$, and $K_2$, $\Gamma$ designed as in Algorithm \ref{tvftalg}.
 \end{theorem}

\begin{proof}
  The condition that (\ref{fescon2}) holds $\forall i\in\mathcal{I}_N\setminus\mathcal{I}_M$ is equivalent to $(A+BK_0)h'_j(t)-\dot{h}'_j(t)=0$, $\forall j\in\{2,\cdots,N-M+1\}$, which means that the TVFT defined by $h'_F=\text{col}(h'_2,h'_3,\cdots,h'_{N-M+1})$ is feasible for the auxiliary multi-agent system (\ref{ydot}). According to Lemma \ref{tvftmllma}, it remains to show that (\ref{tvftmlapij})-(\ref{alphahatpdot}) solves the TVFT for multi-agent system (\ref{ydot}) defined by $h'_F$ with a single leader.

  Extensions of Lemma \ref{lmaXi} and Lemma \ref{lmaQ} apply to $\mathcal{G}'$ and $\hat{\mathcal{G}}'$, and are not repeated for compactness. Let $h'=\text{col}(h'_1,h'_F)$ $\hat{d}'_k(t)=d'_{i_k}(t)-d'_{k+1}(t)$ be the error vector between the parent and the child nodes of the directed edge $\hat{e}'_{i_k,k+1}$, $k\in\mathcal{I}_{N-M}$, and denote $\hat{d}'=\text{col}(\hat{d}'_1,\cdots,\hat{d}'_{N-M+1})$. Then $\hat{d}'=(\Xi'\otimes\textbf{I}_n)d'$. Let $Q'(t)=\tilde{Q}'+\hat{Q}'(t)$ where $\Xi'$ and $\tilde{Q}'$ is defined as in Lemma \ref{lmaXi} and \ref{Q}, respectively, based on $\hat{\mathcal{G}}'$ and
 \begin{align}\label{tvcfqb}
      \hat{Q}'_{kj}(t)=\left\{
                   \begin{array}{ll}
                     \hat{a}'_{j+1,i_j}(t), & \text{if}\quad  j=k,  \\
                     -\hat{a}'_{j+1,i_j}(t), & \text{if}\quad  j=i_k-1,  \\ 
                     0, & \text{otherwise}.
                   \end{array}
                 \right.
 \end{align} where the time-varying weights are defined in (\ref{tvftmlapij}).

 With (\ref{tvftmlui}), the closed-loop state dynamics of the leader-following multi-agent system (\ref{ydot}) can be obtained as
  \begin{align}\label{xdot2}
  \dot{y}=&(\textbf{I}_{N-M+1}\otimes A)y+(\mathcal{L}'(t)\otimes BK_2)d'    \nonumber\\
                  &+(\textbf{I}_{N-M+1}\otimes BK_0)h'.
 \end{align}
 Then, it follows from (\ref{xdot2}) and the definitions of $d$ and $\hat{d}$ that
 \begin{align}\label{dhdot}
  \dot{\hat{d}}'=&(\textbf{I}_{N-M}\otimes A)\hat{d}'+(\Xi'\mathcal{L}'(t)\otimes BK_2)d'    \nonumber\\
                  &+(\Xi'\otimes(A+BK_0))h'-(\Xi'\otimes \textbf{I}_n)\dot{h}'      \nonumber\\
                =&(\textbf{I}_{N-M}\otimes A+Q'(t)\otimes BK_2)\hat{d}'    \nonumber\\
                  &+(\Xi'\otimes(A+BK_0))h'-(\Xi\otimes \textbf{I}_n)\dot{h}'
\end{align}
where $\mathcal{L}'(t)$ is the time-varying Laplacian matrix of $\mathcal{G}'(t)$. Under the feasibility condition (\ref{fescon2}), one has
 \begin{align}\label{dhdot2}
  \dot{\hat{d}}'=(\textbf{I}_{N-M}\otimes A+Q'(t)\otimes BK_2)\hat{d}'.
\end{align}

Consider the Lyapunov candidate as
\begin{align}\label{V2}
  V_2(t)=\frac{1}{2}\hat{d}'^T(&\textbf{I}_{N-M}\otimes P^{-1})\hat{d}'  \nonumber\\
                       & +\sum_{k=1}^{N-M}\frac{1}{2\rho_{k+1,i_k}}(\hat{a}'_{k+1,i_k}(t)-\delta_{k+1,i_k})^2
\end{align}
where $P$ is a solution of (\ref{mainlmi2}) and $\delta_{k+1,i_k}\in\mathbb{R}^{+}$, $k\in\mathcal{I}_{N-M}$. Following similar steps as in the proof of Theorem \ref{tvfthrem}, one has $\lim_{t\rightarrow\infty}\hat{d}'(t)=0$. In this case, the TVFT with a single leader is realized in (\ref{ydot}), meanwhile, the TVFT with multiple leaders is realized in (\ref{leaders}). This completes the proof.
\end{proof}

In the special case when $M=1$, the auxiliary multi-agent system (\ref{ydot}) coincides with the original one, thus, it can be removed. The DST-based adaptive TVFT controller can be directly designed for follower $i$, $i=2,\cdots,N$, as:
 \begin{align}\label{tvftui}
   &u_i=K_0h_i+K_2\sum_{j\in\mathcal{N}_{1}(i)}\alpha_{ij}(t)(d_i-d_j)  \\
   \label{tvftaijui}
   &\alpha_{ij}(t)=\left\{
                    \begin{array}{ll}
                      a_{ij}, & \text{if}\quad e_{ji}\in\mathcal{E}\setminus\hat{\mathcal{E}}, \\
                      \hat{a}_{k+1,i_k}(t), &  \text{if}\quad e_{ji}\in\hat{\mathcal{E}}
                    \end{array}
                  \right.
 \end{align}
 and adaptive laws $\dot{\hat{a}}_{k+1,i_k}$ as in (\ref{alphabardot}). Here, $d_1(t)=x_1(t)$. Immediately, we have the following corollary.
\begin{corollary}[Single leader case]
    Suppose there exists a DST $\hat{\mathcal{G}}$ rooting at the leader. Under feasibility condition (\ref{fescon2}), the TVFT with a single leader is solved by (\ref{tvftui})-(\ref{tvftaijui}) and $\dot{\hat{a}}_{k+1,i_k}$ as in (\ref{alphabardot}), along the designs in Algorithm \ref{tvftalg}.
\end{corollary}

\begin{remark}
    With a single leader, Assumption \ref{dstml} degenerates to the standard assumption of existence of a DST rooting at the leader (\cite{xiao2009finite,dong2018time}, etc). The benefit of Theorem \ref{tvftmlthrem} is thus to provide a natural unifying framework for the DST adaptive method in the presence of one or more leaders.
\end{remark}

 \begin{remark}\label{k0k2}
   The TVFT problem with a single leader can be seen as a special type of the TVF problem where $h_1(\cdot)\equiv0$ for the leader. By comparing (\ref{tvfui}) with (\ref{tvftui}), it can be seen that $K_1=-K_0$ in (\ref{tvftui}). This means that there is no separate term for the average formation signal, since the formation reference is known a prior as the of leader's trajectory.
\end{remark}


\section{Numerical examples}\label{exsec}
In this section, three numerical examples for TVF, TVFT with three leaders and with a single leader are implemented to validate the theoretical results. In all three examples, the initial positions of the agents (followers) are chosen from a Gaussian distribution with standard deviation $5$, and the initial coupling weights of the edges are chosen from a uniform distribution in the interval $(0,0.1)$.
 \begin{example}[TVF]\label{tvfex}
 \begin{figure}[thb!]
  \centering
  \includegraphics[width=0.50\textwidth]{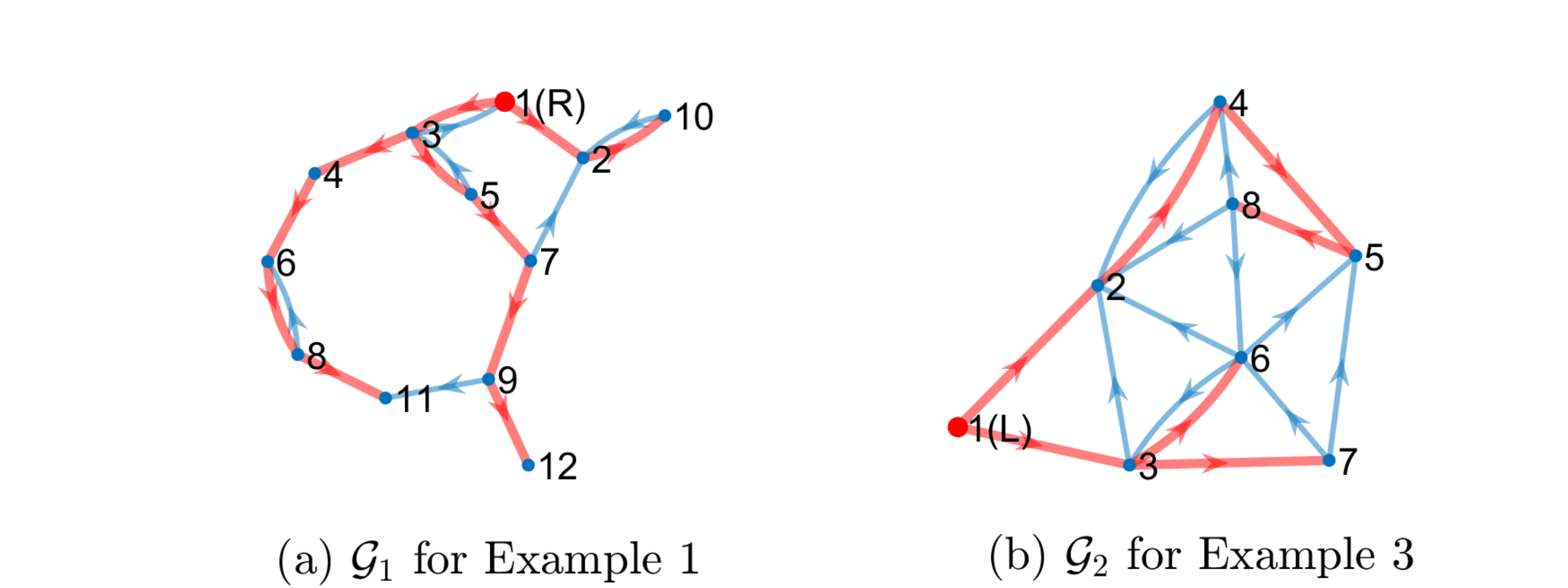}\\
  \caption{Communication graphs. The DSTs are highlighted with red color, and (R), (L) are the root and leader nodes.}\label{tops}
 \end{figure}
    \begin{figure}[thb!]
  \centering
  \includegraphics[width=0.40\textwidth]{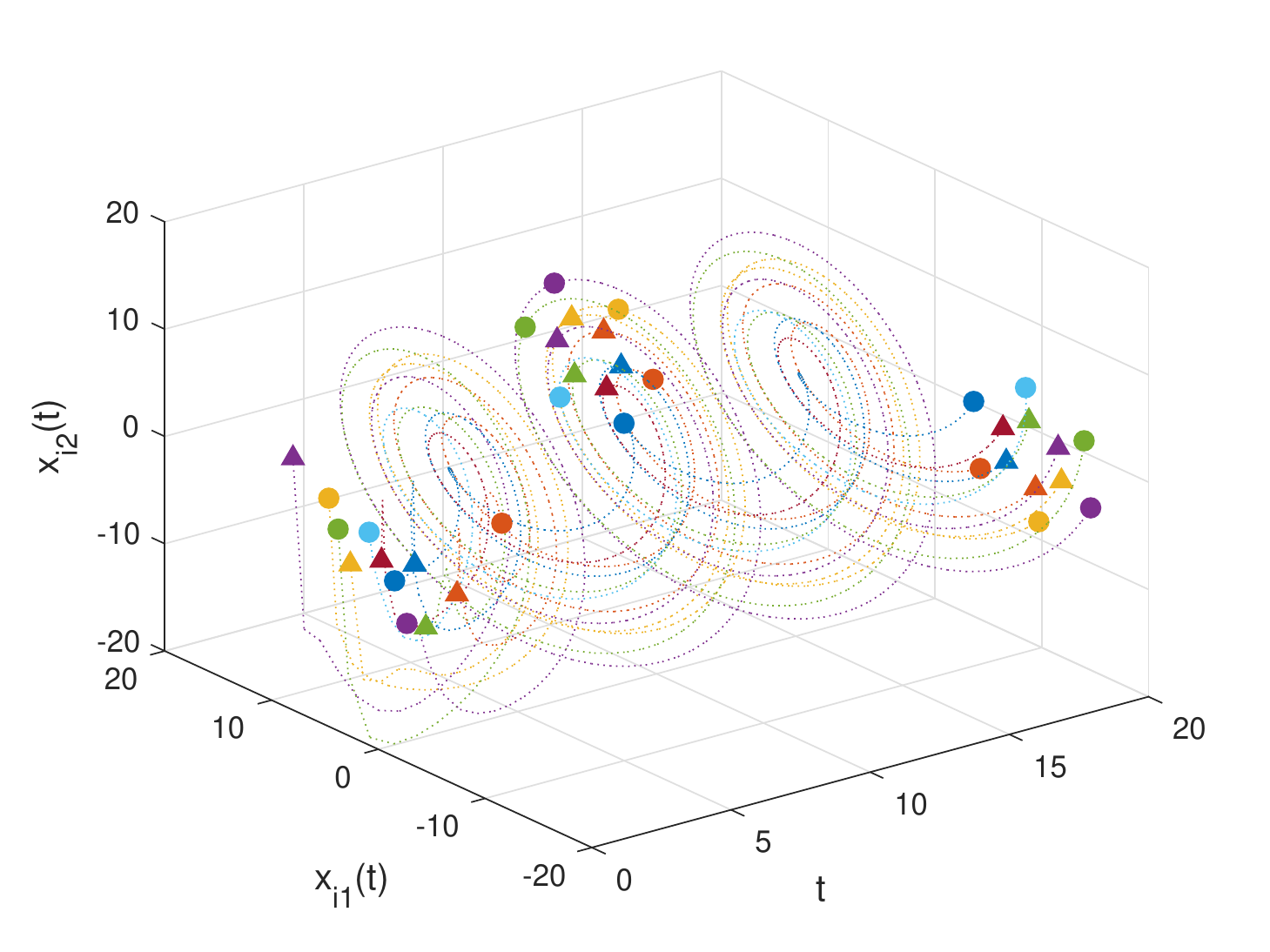}\\
  \caption{Example \ref{tvfex} (TVF): Trajectories of the agents $x_i(t)$, where the circles and triangles are used to mark the agents $i\in\mathcal{I}_6$ and the agents $i\in\mathcal{I}_{12}\setminus\mathcal{I}_6$, respectively, at $t=0$, $10$ and $20$.}\label{e1s1}
 \end{figure}
 \begin{figure}[thb!]
  \centering
  \includegraphics[width=0.50\textwidth]{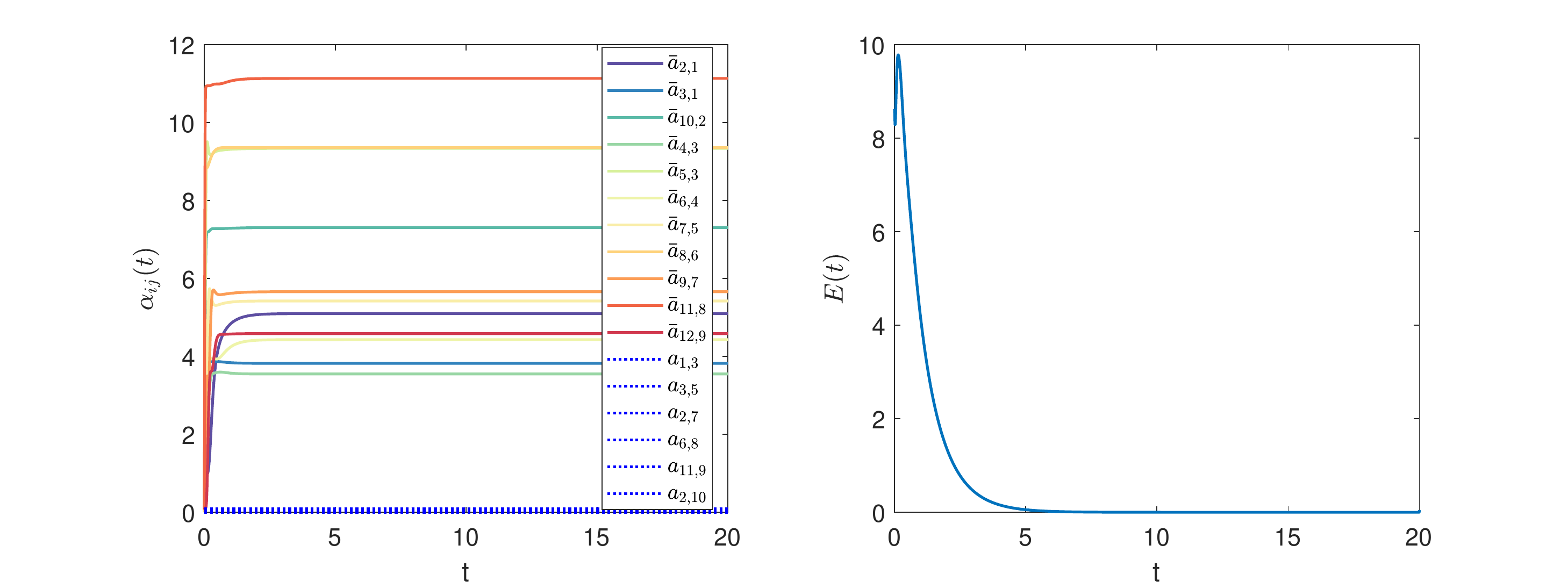}\\
  \caption{Example \ref{tvfex} (TVF): Coupling weights $\alpha_{ij}(t)$ and global formation error $E(t)$ with proposed adaptive method.}\label{e1ae1}
 \end{figure}
  \begin{figure}[thb!]
  \centering
  \includegraphics[width=0.50\textwidth]{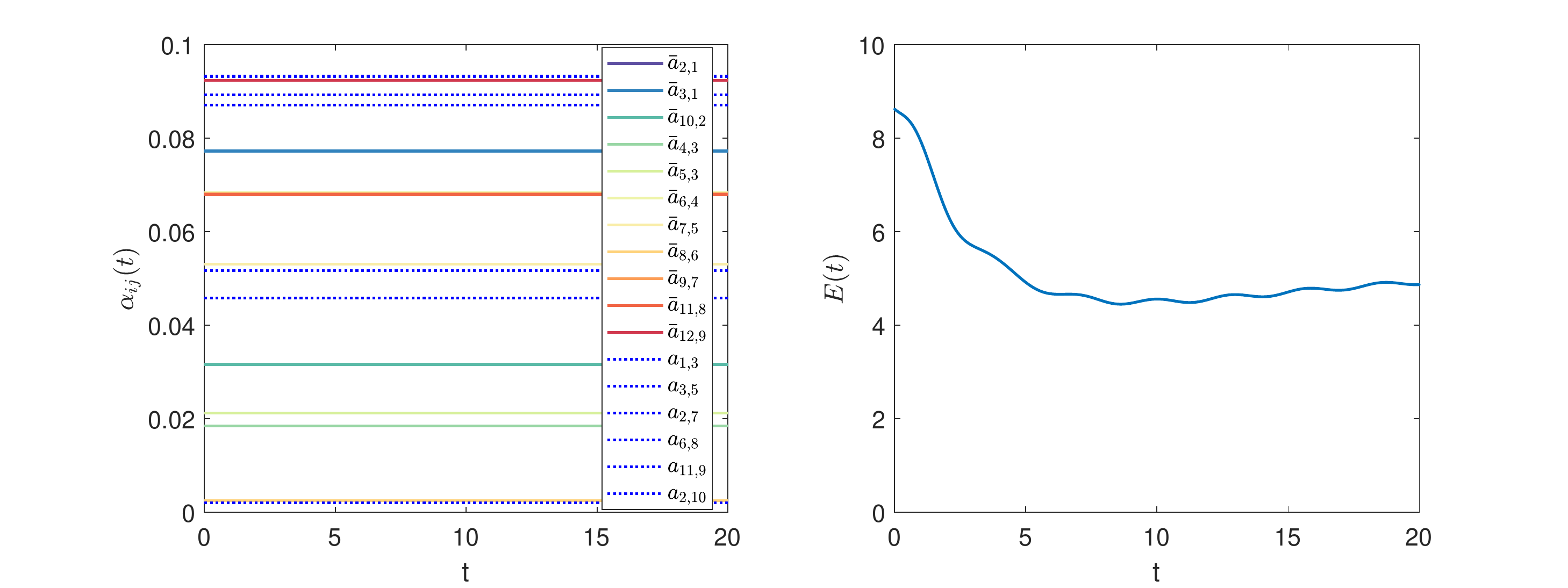}\\
  \caption{Example \ref{tvfex} (TVF): Coupling weights $\alpha_{ij}$ and global formation error $E(t)$ with nonadaptive adaptive method (same initial $\alpha_{ij}$ as in Fig.~\ref{e1ae1}).}\label{e1ae2}
 \end{figure}
   Consider a second-order system modelled by (\ref{agents}) with $N=12$,
$
 A=\left(
      \begin{array}{cc}
        0 & 1 \\
        -1 & 2 \\
      \end{array}
    \right),
  B=\left(
      \begin{array}{c}
        0 \\
        1 \\
      \end{array}
    \right).
$

  The agents interact on the digraph $\mathcal{G}_1$ in Fig.~\ref{tops}. The required TVF is a pair of nested hexagons with $h_i(t)=(6\sin(t+\frac{(i-1)\pi}{3}),6\cos(t+\frac{(i-1)\pi}{3}))^T$ for $i\in\mathcal{I}_6$, and $h_i(t)=(3\sin(t+\frac{(i-1)\pi}{3}),3\cos(t+\frac{(i-1)\pi}{3}))^T$ for $i\in\mathcal{I}_{12}\setminus\mathcal{I}_6$.

  Let $K_0=(0,-2)$. It can be verified via condition (\ref{fescon}) that the desired formation is feasible for the selected DST. Since $A+BK_0$ is stabilizable, we can assign $K_1=(0,0)$. Let $\eta=2$, $\theta=1$, and solve LMI (\ref{mainlmi}) to give a solution
  $P=\left(
     \begin{array}{cc}
       1/3 & -1/3  \\
       -1/3 & 2/3  \\
     \end{array}
   \right)$. Following Algorithm \ref{tvfalg}, one has $K_2=(-3,-3)$, and
 $\Gamma=\left(
     \begin{array}{cc}
       9 & 9  \\
       9 & 9  \\
     \end{array}
   \right)$. Let $\rho_{k+1,i_k}=0.1$.

 The trajectories of the agents are in Fig.~\ref{e1s1}, showing how the nested hexagons are formed and rotate. Let $e_i(t)=d_i(t)-d_{\text{ave}}$ (see Remark \ref{k0k1k2}), $i\in\mathcal{I}_{N}$. The global formation error $E(t)=\sqrt{\frac{1}{N}\sum_{i=1}^{N}\|e_i(t)\|^2}$ converges to zero, as shown in Fig.~\ref{e1ae1}. Fig.~\ref{e1ae1} also shows that the weights $\alpha_{ij}$ are time-varying on the DST (solid lines) and kept constant otherwise (dashed lines). For comparison, Fig.~\ref{e1ae2} shows that if all weights are kept constant \big($\alpha_{ij}=\alpha_{ij}(0)$\big), no TVF may be achieved (global formation error does not converge to zero).
 \end{example}

\begin{example}[TVFT with Three Leaders]\label{tvftmlex}
  \begin{figure}[thb!]
  \centering
  \includegraphics[width=0.40\textwidth]{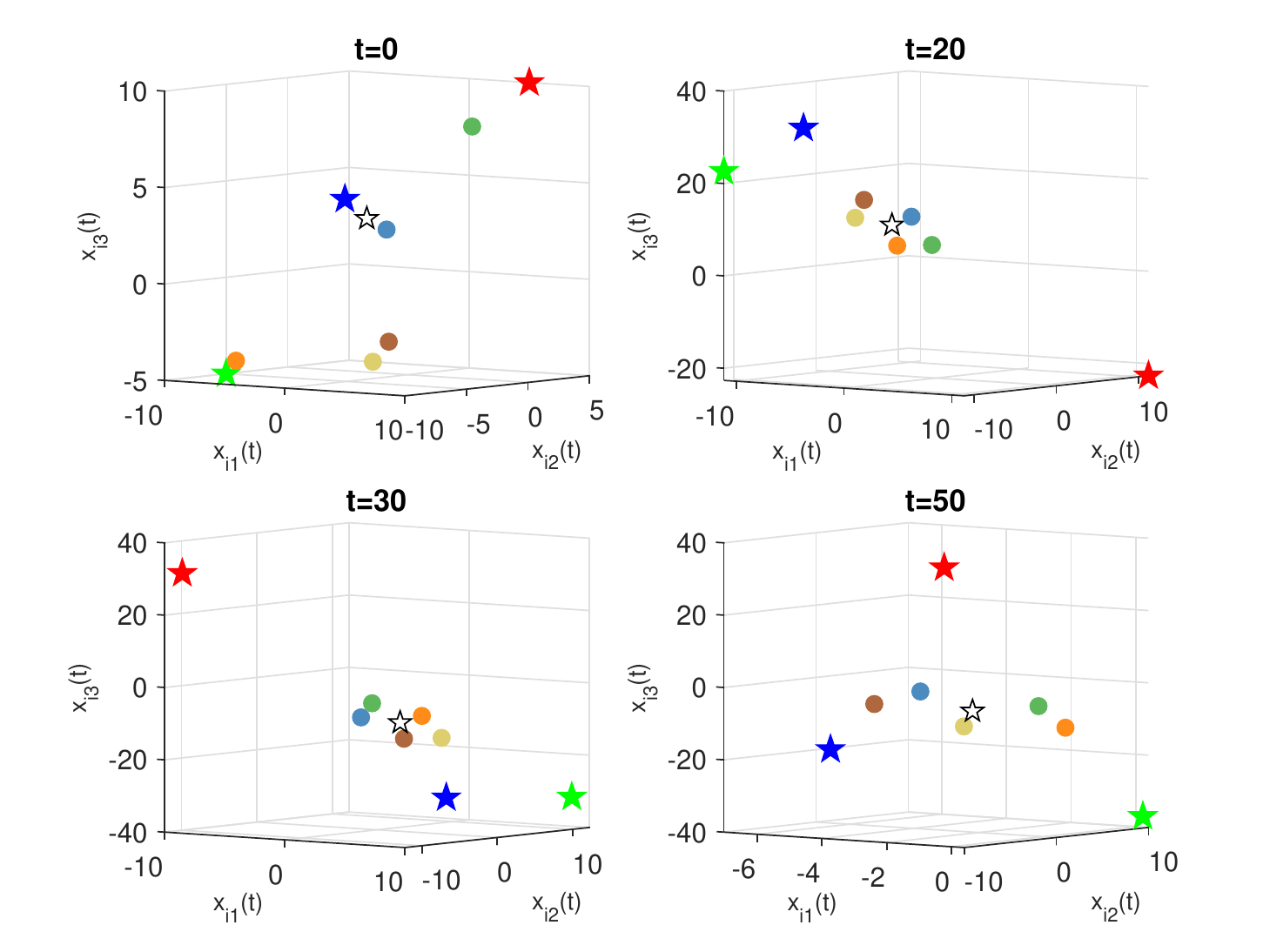}\\
  \caption{Example \ref{tvftmlex} (TVFT with Three Leaders): Snapshots at $t=0$, $20$, $30$, and $50$. Three filled pentagrams, five circles and an unfilled pentagram are used to mark leaders, followers, and the average of the leaders, respectively.}\label{e3s1}
 \end{figure}
   \begin{figure}[thb!]
  \centering
  \includegraphics[width=0.50\textwidth]{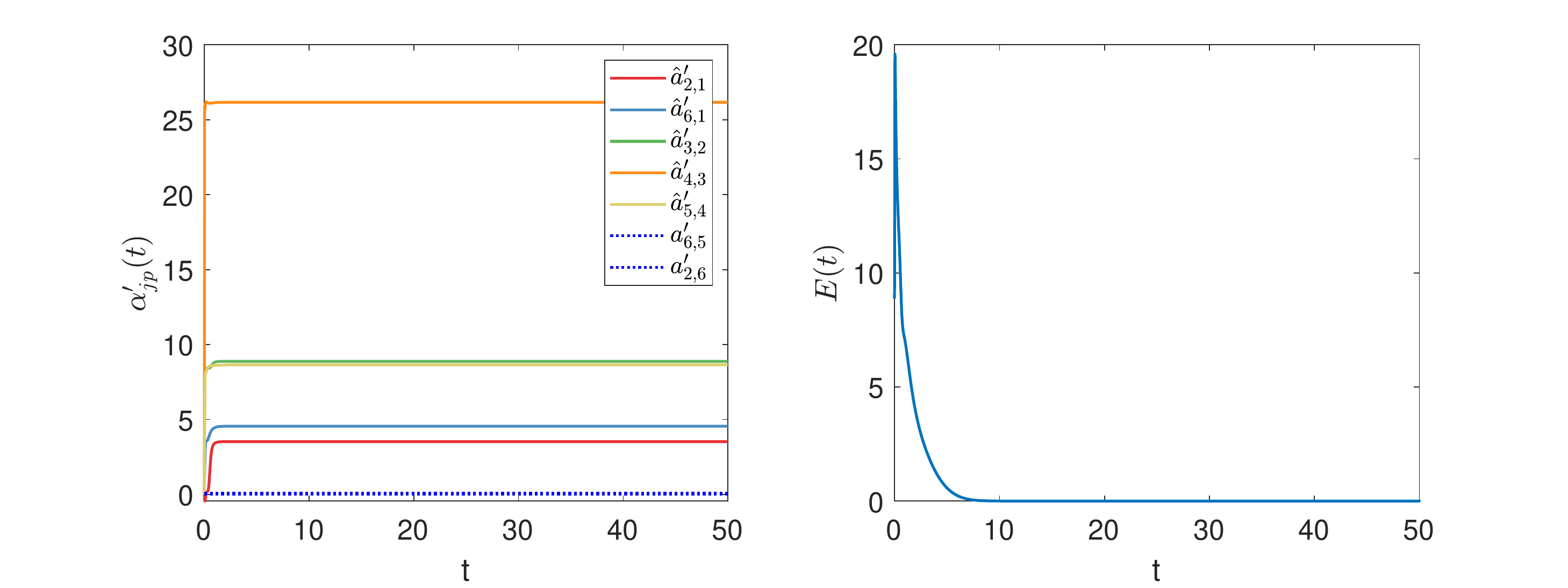}\\
  \caption{Example \ref{tvftmlex} (TVFT with Three Leaders): Coupling weights $\alpha'_{jp}(t)$ in $\mathcal{G}'$, and global formation tracking error $E(t)$ with proposed adaptive method.}\label{e3ae1}
 \end{figure}
    \begin{figure}[thb!]
  \centering
  \includegraphics[width=0.50\textwidth]{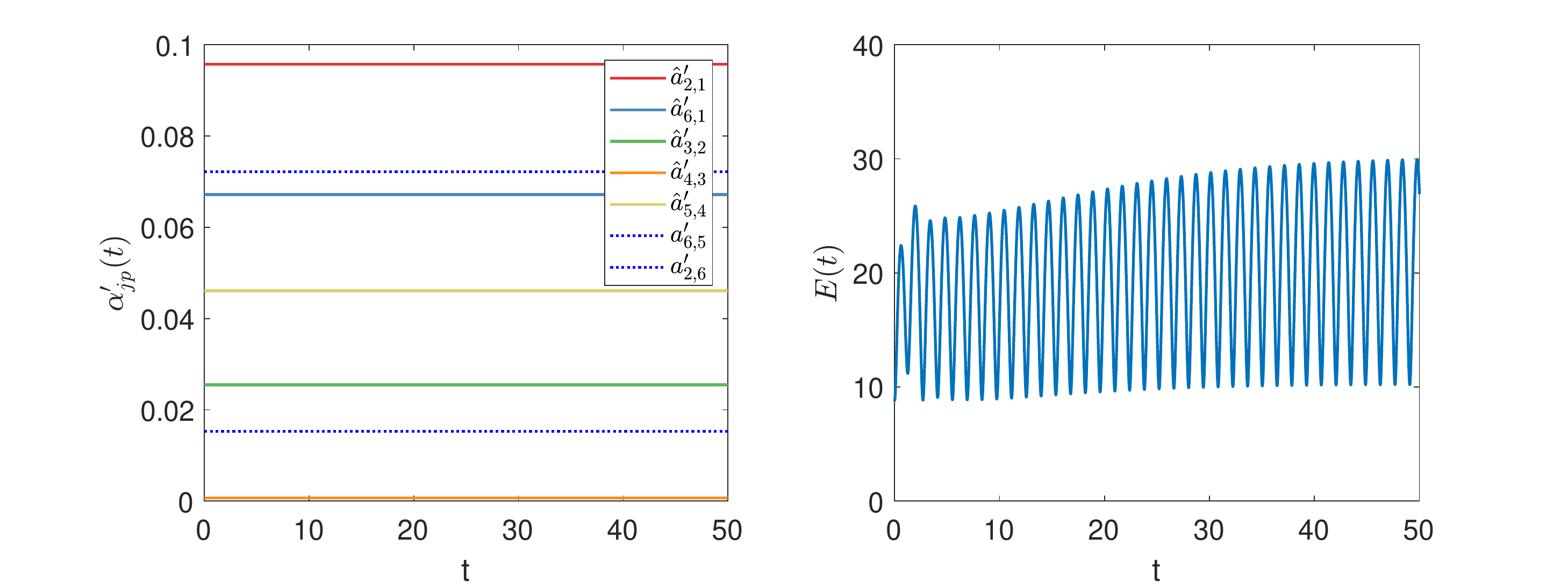}\\
  \caption{Example \ref{tvftmlex} (TVFT with Three Leaders): Coupling weights $\alpha'_{jp}$ in $\mathcal{G}'$, and global formation tracking error $E(t)$ with nonadaptive control (same initial $\alpha'_{jp}$ as in Fig.~\ref{e3ae1}).}\label{e3ae2}
 \end{figure}
  Consider a third-order multi-agent system modelled by (\ref{leaders}) with $N=8$, $M=3$, and
$$
 A=\left(
      \begin{array}{ccc}
        0 & 1 & 1 \\
        1 & 2 & 1 \\
        -2 &-10 & -3 \\
      \end{array}
    \right),
  B=\left(
      \begin{array}{c}
        0 \\
        0 \\
        1 \\
      \end{array}
    \right).
$$

The communication graph is the digraph $\mathcal{G}$ in Fig.~\ref{topex3}. The followers are required to form a time-varying pentagram described by
$$
 h_i(t)=\left(
      \begin{array}{c}
        3\sin(t+\frac{2(i-4)\pi}{5}) \\
        -3\cos(t+\frac{2(i-4)\pi}{5}) \\
        6\cos(t+\frac{2(i-4)\pi}{5}) \\
      \end{array}
    \right), \quad i=4,5\cdots,8,
$$ while tracking the average of the states of the leaders, i.e., $\beta_1=\beta_2=\beta_3=1/3$.

Let $K_0=(0,4,0)$. It can be verified that the defined $h_i(\cdot)$ is feasible. Let $\eta=2$, $\theta=1$, and $\rho_{k+1,j_k}=0.1$. Following Algorithm \ref{tvftalg}, one has $K_2=(-2.3066,-6.8257,-2.4970)$, and
 $\Gamma=\left(
     \begin{array}{ccc}
       5.3206 & 15.7444 & 5.7596 \\
       15.7444 & 46.5895 & 17.0434  \\
         5.7596 & 17.0434  & 6.2349
     \end{array}
   \right)$.

The initial value of the leaders are chosen as $x_1(0)=(5,5,10)^T$, $x_2(0)=(-10,-5,-5)^T$, $x_3(0)=(5,-10,5)^T$. Several snapshots of the agents are in Fig.~\ref{e3s1}, showing that the pentagram emerges and rotates around the average of the three leaders. Similarly, we define the global formation tracking error $E(t)=\sqrt{\frac{1}{N-3}\sum_{i=4}^{N}\|d_i(t)-\sum_{l=1}^3\beta_lx_l(t)\|^2}$. The trajectories of $\alpha'_{ij}$ in $\mathcal{G}'$ (see Fig.~\ref{topex3}) and $E(t)$ are provided in Fig.~\ref{e3ae1}. Once more, a constant coupling strategy fails to accomplish the TVFT task, as shown in Fig.~\ref{e3ae2}.
\end{example}

 \begin{example}[TVFT with a Single Leader]\label{tvftexsl}
  \begin{figure}[thb!]
  \centering
  \includegraphics[width=0.40\textwidth]{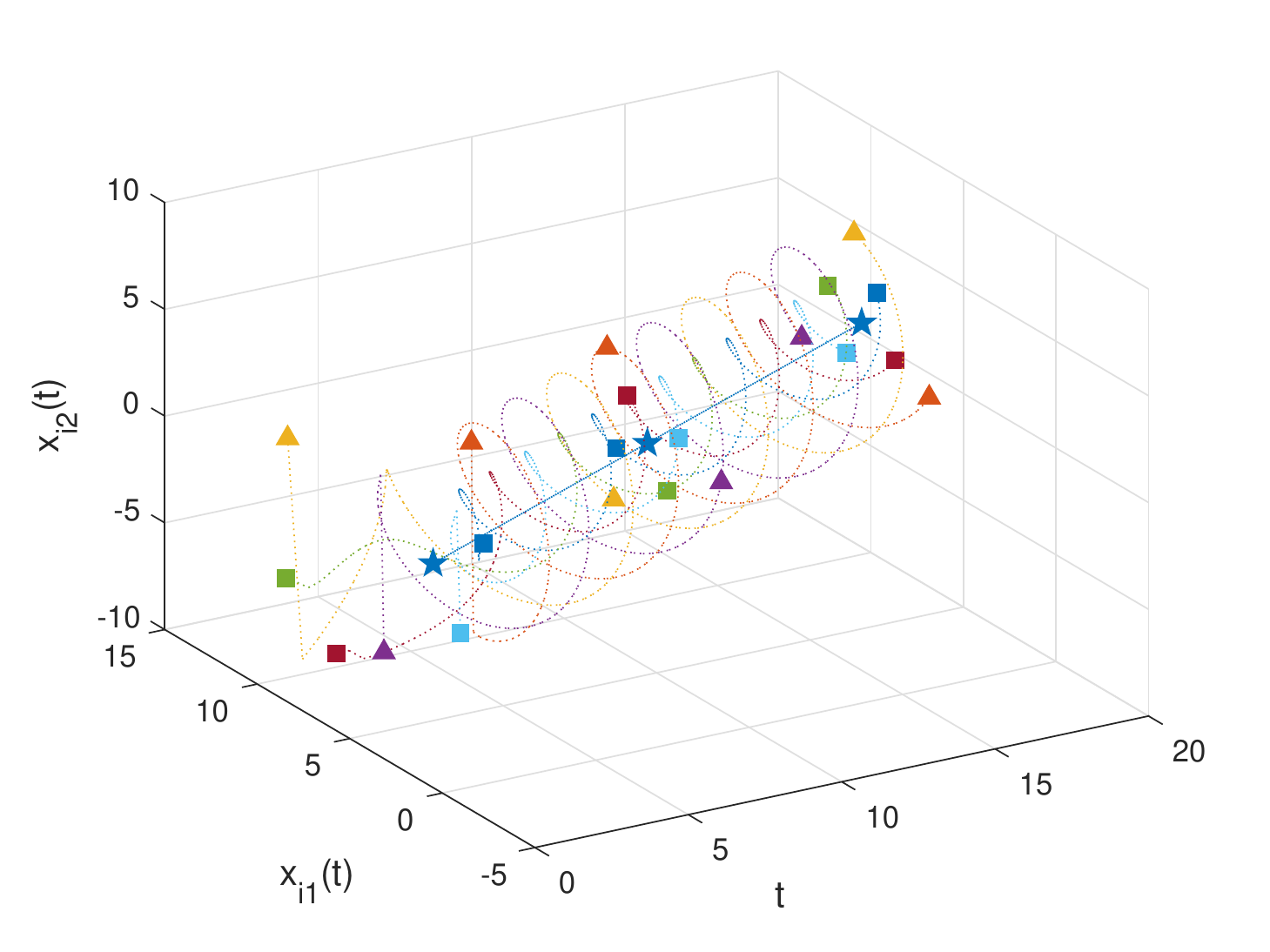}\\
  \caption{Example \ref{tvftexsl} (TVFT with a Single Leader): Trajectories of the agents $x_i(t)$, where three triangles, four squares and a pentagram are used to mark the agents $i\in\{2,3,4\}$, $i\in\{5,6,7,8\}$, and the leader $i=1$, respectively, at $t=0$, $10$ and $20$.}\label{e2s1}
 \end{figure}
  \begin{figure}[thb!]
  \centering
  \includegraphics[width=0.50\textwidth]{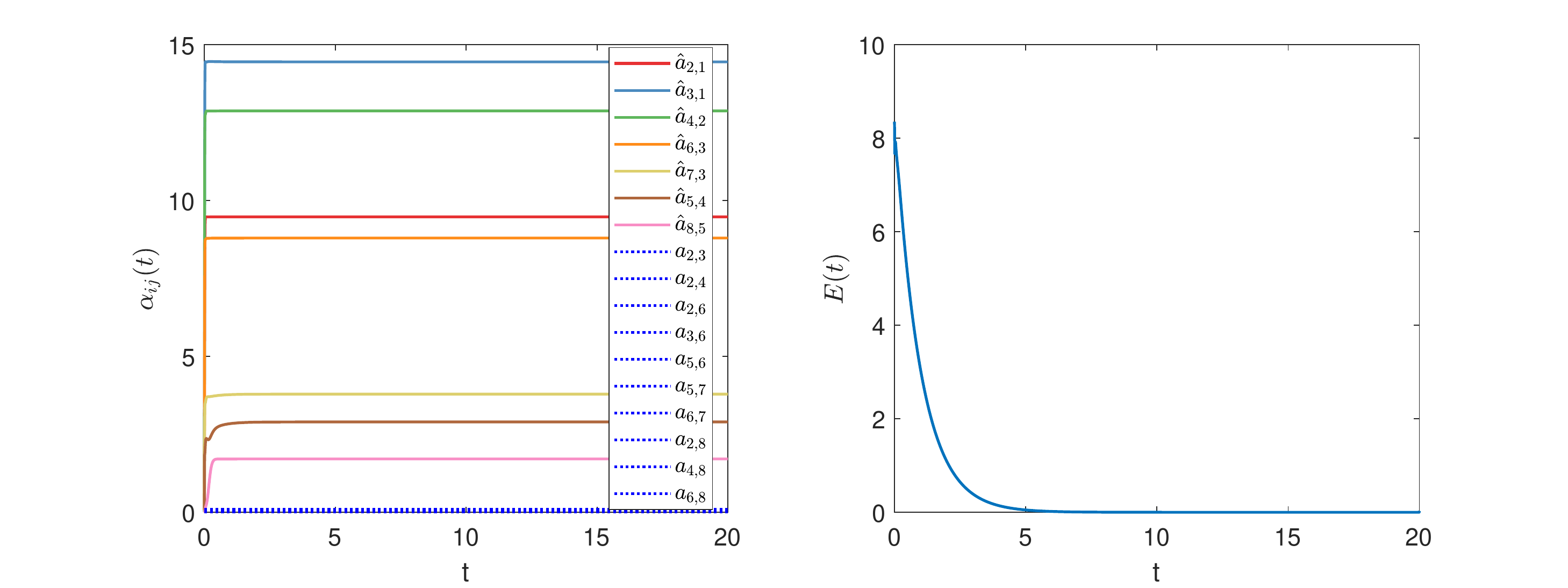}\\
  \caption{Example \ref{tvftexsl} (TVFT with a Single Leader): Coupling weights $\alpha_{ij}(t)$ and global formation tracking error $E(t)$ with proposed adaptive method.}\label{e2ae1}
 \end{figure}
   \begin{figure}[thb!]
  \centering
  \includegraphics[width=0.50\textwidth]{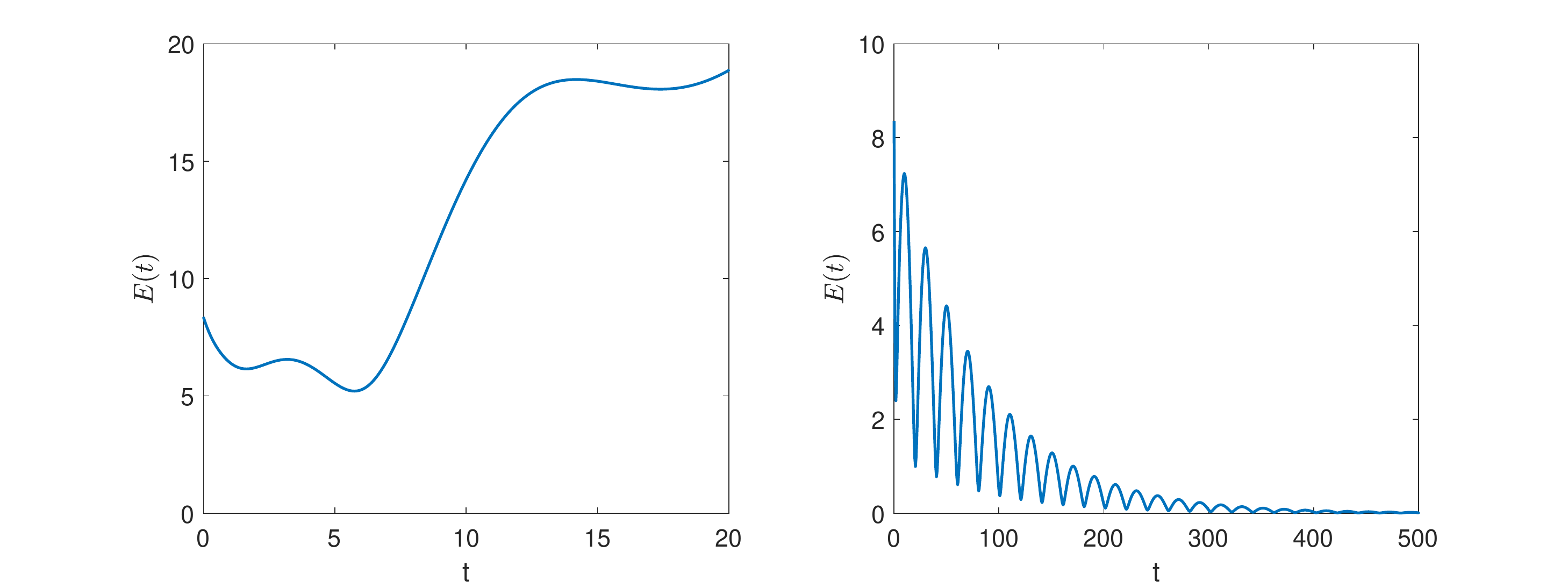}\\
  \caption{Example \ref{tvftexsl} (TVFT with a Single Leader): Global formation tracking error $E(t)$ with nonadaptive control (left) and with adaptive controller (\ref{cdot}) (right).}\label{e2ee}
 \end{figure}
 Consider a network of second-order agents with $N=8$, $M=1$,
$
 A=\left(
      \begin{array}{cc}
        0 & 1 \\
        0 & 0 \\
      \end{array}
    \right),
  B=\left(
      \begin{array}{c}
        0 \\
        1 \\
      \end{array}
    \right),
$ and digraph $\mathcal{G}_2$ in Fig.~\ref{tops}.

The desired formation is an equilateral triangle-like formation around the leader, which is specified by $h_i(t)=(4\sin(t+\frac{2(i-2)\pi}{3}+\pi),4\cos(t+\frac{2(i-2)\pi}{3}+\pi))^T$ for $i\in\{2,3,4\}$, and $h_i(t)=(2\sin(t+\frac{(i-5)\pi}{2}),2\cos(t+\frac{(i-5)\pi}{2}))^T$ for $i\in\{5,6,7,8\}$.

Let $K_0=(-1,0)$. It can be verified via condition (\ref{fescon2}) that the desired formation is feasible. Let $\eta=2$, $\theta=1$, and solve the LMI (\ref{mainlmi2}) to give a solution
$P=\left(
     \begin{array}{cc}
       0.6513 & -0.6513  \\
       -0.6513 & 0.8256  \\
     \end{array}
   \right)$. Following Algorithm \ref{tvftalg}, one has $K_2=(-5.7356,-5.7356)$, and
 $\Gamma=\left(
     \begin{array}{cc}
       32.8969 & 32.8969  \\
       32.8969 & 32.8969  \\
     \end{array}
   \right)$. We choose $\rho_{k+1,i_k}=0.1$.

The initial value of the leader is chosen as $x_1(0)=(0.5,0.5)^T$. The trajectories of the agents are in Fig.~\ref{e2s1}, showing how the triangle emerges and rotates around the leader. If we define the global formation tracking error as $E(t)=\sqrt{\frac{1}{N-1}\sum_{i=2}^{N}\|d_i(t)-x_1(t)\|^2}$, we can see from Fig.~\ref{e2ae1} that it converges to zero (see also the time-varying weights $\alpha_{ij}$ on the DST). Fig.~\ref{e2ee} (left) shows that also in this case the TVFT may not be achieved with nonadaptive control.

The DST framework is not the only possible framework to remove the knowledge of the Laplacian eigenvalues: alternative frameworks have been proposed for consensus \cite{lv2017novel} and group TVFT \cite{hu2020distributed}. Let us include a comparison with the adaptive method used in \cite{lv2017novel,hu2020distributed}, which can be written as:
  \begin{align}\label{cdot}
   &u_i=K_0h_i+K_2\sum_{j\in\mathcal{N}_{1}(i)}(c_{i}(t)+\xi_iP^{-1}\xi_i)(d_i-d_j)   \nonumber\\
   &\dot{c}_i=\xi_i^T\Gamma\xi_i \qquad\xi_i=\sum_{j\in\mathcal{N}_{1}(i)}a_{ij}(d_i-d_j).
 \end{align}
Note that in (\ref{cdot}) all coupling weights in the network are made adaptive. We select the same initial conditions, and $c_i(0)=10$; the global formation tracking error is shown in Fig.~\ref{e2ee} (right). As compared to Fig.~\ref{e2ae1} (right), it is interesting to note that adapting the gains on a DST instead of on the entire network leads to faster convergence of the formation errors.
\end{example}

\section{Conclusions}\label{consec}
A directed spanning tree (DST) adaptive framework has been developed for time-varying formation and formation tracking of linear multi-agent systems. The proposed framework provides a natural generalization of the DST based adaptive method in the presence of one or more leaders: necessary and sufficient conditions for solving the proposed framework have been derived. Future topics may include generalizing the proposed DST framework in the sense of cluster formation, partial state information, nonlinear agents and nonzero inputs of the leaders.

\appendix[Proof of Lemma \ref{lmaQ}]
Inspired by \cite{yu2015distributed} and \cite{yu2018distributed}, an auxiliary matrix $J$ is introduced to analyze Lemma \ref{lmaQ}. Define $J\in\mathbb{R}^{N\times(N-1)}$ as
\begin{align}\label{J}
    J_{ik}=\left\{
              \begin{array}{ll}
                0, & \text{if}\quad i\in\bar{\mathcal{V}}_{k+1},  \\
                1, & \text{otherwise}
              \end{array}
            \right.\nonumber
 \end{align}
 where $\bar{\mathcal{V}}_{k+1}$ represents the vertex set of the subtree of $\bar{\mathcal{G}}$ rooting at node $k+1$. The proof will proceed along three steps:
 \begin{enumerate}
   \item\label{L} Proving that $\mathcal{L}=\mathcal{L}J\Xi$;
   \item\label{Q} Proving that $Q=\Xi\mathcal{L}J$;
   \item\label{XiL} Proving (\ref{xlqx}) and (\ref{qbar}), i.e., the statements of the lemma.
 \end{enumerate}

 Step \ref{L}) Let us denote $X=J\Xi$. Then, $X_{ij}=\sum_{k=1}^{N-1}J_{ik}\Xi_{kj}$, $i,j\in\mathcal{I}_N$. We classify the discussions according to the value of $j$ in order to clarify the matrix $X$.

\emph{Case 1: $j=1$.} Then, $X_{i1}=\sum_{k=1,i_k=1}^{N-1}J_{ik}$.

Since $J_{1k}=1$, $\forall k$, then $X_{11}=\bar{\mathcal{D}}_2(1)$, which is the out-degree of the root in $\bar{\mathcal{G}}$; When $i>1$, there exists a unique $\bar{k}\in\mathcal{I}_{N-1}$ satisfying $i_{\bar{k}}=1$, such that $i\in\bar{\mathcal{V}}_{\bar{k}+1}$, implying that $J_{i\bar{k}}=0$. Thus, $X_{i1}=\bar{\mathcal{D}}_2(1)-1$.

To sum up, $X_{i1}=\left\{
              \begin{array}{ll}
                \bar{\mathcal{D}}_2(1), & i=1,  \\
                \bar{\mathcal{D}}_2(1)-1, & i>1.
              \end{array}
            \right.$

\emph{Case 2: $j$ is a stem.} Then, $X_{ij}=\sum_{k=1,i_k=j}^{N-1}J_{ik}-J_{i,j-1}$.
\begin{enumerate}[i.]
  \item When $i\notin\bar{\mathcal{V}}_j$, $X_{ij}=\sum_{k=1,i_k=j}^{N-1}J_{ik}-1$. Then, $\forall k$ satisfying $i_k=j$, $i\notin\bar{\mathcal{V}}_{k+1}$. Thus, $X_{ij}=\bar{\mathcal{D}}_2(j)-1$.
  \item When $i\in\bar{\mathcal{V}}_j$, $X_{ij}=\sum_{k=1,i_k=j}^{N-1}J_{ik}$. If $i=j$, then $\forall k$ satisfying $i_k=j$, $J_{ik}=1$. Thus $X_{jj}=\bar{\mathcal{D}}_2(j)$. If $i\neq j$, there exists a unique $\bar{k}$ satisfying $i_{\bar{k}}=j$, such that $i\in\bar{\mathcal{V}}_{\bar{k}+1}$, implying that $J_{i\bar{k}}=0$. Then, $X_{ij}=\bar{\mathcal{D}}_2(j)-1$.
\end{enumerate}

To sum up, $X_{ij}=\left\{
              \begin{array}{ll}
                \bar{\mathcal{D}}_2(j), & i=j,  \\
                \bar{\mathcal{D}}_2(j)-1, & i\neq j
              \end{array}
            \right.$ when $j$ is a stem.

\emph{Case 3: $j$ is a leaf.} Then, $X_{ij}=-J_{i,j-1}$.

In this case, $\bar{\mathcal{V}}_j=\{j\}$, meaning that $J_{i,j-1}=0$ if and only if $i=j$. Then $X_{ij}=\left\{
              \begin{array}{ll}
                0, & i=j,  \\
                -1, & i\neq j.
              \end{array}
            \right.$

Summarizing all three cases, the matrix $X$ can be written in a unified way as
$X_{ij}=\left\{
              \begin{array}{ll}
                \bar{\mathcal{D}}_2(j), & i=j,  \\
                \bar{\mathcal{D}}_2(j)-1, & i\neq j.
              \end{array}
            \right.$ Then,
\begin{align}
  (\mathcal{L}X)_{ij}&=\sum_{k=1}^N\mathcal{L}_{ik}X_{kj}\nonumber\\
  &=\sum_{k\neq j}\mathcal{L}_{ik}(\bar{\mathcal{D}}_2(j)-1)+\mathcal{L}_{ij}\bar{\mathcal{D}}_2(j)\nonumber\\
  &=(\bar{\mathcal{D}}_2(j)-1)\sum_{k=1}^N\mathcal{L}_{ik}+\mathcal{L}_{ij}
  =\mathcal{L}_{ij}.\nonumber
\end{align}
So, $\mathcal{L}=\mathcal{L}J\Xi$ is proved.

Step \ref{Q}) Let us denote $Y=\Xi\mathcal{L}J$. Then,
\begin{align}
  Y_{kj}&=\sum_{i=1}^N(\Xi\mathcal{L})_{ki}J_{ij}=\sum_{i=1}^N(\sum_{s=1}^N\Xi_{ks}\mathcal{L}_{si})J_{ij}\nonumber\\
  &=\sum_{s=1}^N\Xi_{ks}\sum_{i=1}^N\mathcal{L}_{si}J_{ij}=\sum_{i=1}^N\mathcal{L}_{i_k,i}J_{ij}-\sum_{i=1}^N\mathcal{L}_{k+1,i}J_{ij}\nonumber\\
  &=\sum_{i=1,i\notin\bar{\mathcal{V}}_{k+1}}^N(\mathcal{L}_{i_k,i}-\mathcal{L}_{k+1,i})\nonumber
\end{align}
where the definitions of $\Xi$ and $J$ are used to get the last two equalities, respectively. Since $\mathcal{L}$ has zero row sums, we have
\begin{align}
  Y_{kj}&=\sum_{c\in\bar{\mathcal{V}}_{k+1}}(\mathcal{L}_{k+1,c}-\mathcal{L}_{i_k,c})\nonumber\\
  &=\sum_{c\in\bar{\mathcal{V}}_{j+1}}(\tilde{\mathcal{L}}_{k+1,c}-\tilde{\mathcal{L}}_{i_k,c})+\sum_{c\in\bar{\mathcal{V}}_{j+1}}(\bar{\mathcal{L}}_{k+1,c}-\bar{\mathcal{L}}_{i_k,c})\nonumber\\
  &=\tilde{Q}_{kj}+\bar{Q}_{kj}=Q_{kj}.\nonumber
\end{align}
Then, $Q=\Xi\mathcal{L}J$ is proved.

Step \ref{XiL}) Let both sides $Q=\Xi\mathcal{L}J$ multiply $\Xi$, one has $Q\Xi=\Xi\mathcal{L}J\Xi=\Xi\mathcal{L}$, then (\ref{xlqx}) holds. To prove the explicit form of $\bar{Q}$ in (\ref{qbar}), one can can distinguish three cases based on the relationships between the edge $\bar{e}_{i_k,k+1}$ and the subtree $\bar{\mathcal{V}}_{j+1}$:

\emph{Case 1: $k+1\notin\bar{\mathcal{V}}_{j+1}$}. Then, it is obvious that $\bar{Q}_{kj}=0$.

\emph{Case 2: $k+1\in\bar{\mathcal{V}}_{j+1}$} and \emph{$i_k\notin\bar{\mathcal{V}}_{j+1}$}. In this case, the only possible value of $k$ is $k=j$. Then,
\begin{align}
  \bar{Q}_{kj}&=\sum_{c\in\bar{\mathcal{V}}_{j+1}}(\bar{\mathcal{L}}_{k+1,c}-\bar{\mathcal{L}}_{i_k,c})\nonumber\\
  &=\bar{\mathcal{L}}_{k+1,k+1}=\bar{\mathcal{L}}_{j+1,j+1}=\bar{a}_{j+1,i_j}. \nonumber
\end{align}

\emph{Case 3: $i_k\in\bar{\mathcal{V}}_{j+1}$}. Then,
\begin{enumerate}[i.]
  \item When $i_k=j+1$,
  \begin{align}
  \bar{Q}_{kj}&=\sum_{c\in\bar{\mathcal{V}}_{j+1}}(\bar{\mathcal{L}}_{k+1,c}-\bar{\mathcal{L}}_{i_k,c})\nonumber\\
  &=\bar{\mathcal{L}}_{k+1,i_k}-\bar{\mathcal{L}}_{i_k,i_k}+\bar{\mathcal{L}}_{k+1,k+1}-\bar{\mathcal{L}}_{i_k,k+1} \nonumber\\
  &=-\bar{\mathcal{L}}_{i_k,i_k}=-\bar{a}_{j+1,i_j}.\nonumber
\end{align}
  \item When $i_k>j+1$,
    \begin{align}
  \bar{Q}_{kj}&=\sum_{c\in\bar{\mathcal{V}}_{j+1}}(\bar{\mathcal{L}}_{k+1,c}-\bar{\mathcal{L}}_{i_k,c})\nonumber\\
  &=\bar{\mathcal{L}}_{k+1,i_{i_k-1}}-\bar{\mathcal{L}}_{i_k,i_{i_k-1}}+\bar{\mathcal{L}}_{k+1,i_k}-\bar{\mathcal{L}}_{i_k,i_k}\nonumber\\
  &\qquad+\bar{\mathcal{L}}_{k+1,k+1}-\bar{\mathcal{L}}_{i_k,k+1} \nonumber\\
  &=-\bar{\mathcal{L}}_{i_k,i_{i_k-1}}+\bar{\mathcal{L}}_{k+1,i_k}-\bar{\mathcal{L}}_{i_k,i_k}+\bar{\mathcal{L}}_{k+1,k+1}
  =0.\nonumber
\end{align}
\end{enumerate}

Summarizing all three cases, the matrix $\bar{Q}$ can also be given in a unified way as
$\bar{Q}_{kj}=\left\{
                   \begin{array}{ll}
                     \bar{a}_{j+1,i_j}, & \text{if}\quad  j=k,  \\
                     -\bar{a}_{j+1,i_j}, & \text{if}\quad  j=i_k-1,  \\ 
                     0, & \text{otherwise}.
                   \end{array}
                 \right.    \nonumber
$ Then (\ref{qbar}) is proved, which completes the proof.




%



\ifCLASSOPTIONcaptionsoff
  \newpage
\fi



\bibliographystyle{IEEEtran}
\bibliography{IEEEabrv,sample}

\begin{thebibliography}{10}
\providecommand{\url}[1]{#1}
\csname url@samestyle\endcsname
\providecommand{\newblock}{\relax}
\providecommand{\bibinfo}[2]{#2}
\providecommand{\BIBentrySTDinterwordspacing}{\spaceskip=0pt\relax}
\providecommand{\BIBentryALTinterwordstretchfactor}{4}
\providecommand{\BIBentryALTinterwordspacing}{\spaceskip=\fontdimen2\font plus
\BIBentryALTinterwordstretchfactor\fontdimen3\font minus
  \fontdimen4\font\relax}
\providecommand{\BIBforeignlanguage}[2]{{%
\expandafter\ifx\csname l@#1\endcsname\relax
\typeout{** WARNING: IEEEtran.bst: No hyphenation pattern has been}%
\typeout{** loaded for the language `#1'. Using the pattern for}%
\typeout{** the default language instead.}%
\else
\language=\csname l@#1\endcsname
\fi
#2}}
\providecommand{\BIBdecl}{\relax}
\BIBdecl

\bibitem{oh2015survey}
K.-K. Oh, M.-C. Park, and H.-S. Ahn, ``A survey of multi-agent formation
  control,'' \emph{Automatica}, vol.~53, pp. 424--440, 2015.

\bibitem{hu2020cooperative}
J.~{Hu}, P.~{Bhowmick}, F.~{Arvin}, A.~{Lanzon}, and B.~{Lennox}, ``Cooperative
  control of heterogeneous connected vehicle platoons: An adaptive
  leader-following approach,'' \emph{IEEE Robot. Autom. Lett.}, vol.~5, no.~2,
  pp. 976--983, 2020.

\bibitem{liu2012iterative}
Y.~Liu and Y.~Jia, ``An iterative learning approach to formation control of
  multi-agent systems,'' \emph{Syst. Control Lett.}, vol.~61, no.~1, pp.
  148--154, 2012.

\bibitem{brinon2014cooperative}
L.~Brin{\'o}n-Arranz, A.~Seuret, and C.~Canudas-de Wit, ``Cooperative control
  design for time-varying formations of multi-agent systems,'' \emph{IEEE
  Trans. Autom. Control}, vol.~59, no.~8, pp. 2283--2288, 2014.

\bibitem{dong2014formation}
X.~Dong, J.~Xi, G.~Lu, and Y.~Zhong, ``Formation control for high-order linear
  time-invariant multiagent systems with time delays,'' \emph{IEEE Trans.
  Control Netw. Syst.}, vol.~1, no.~3, pp. 232--240, 2014.

\bibitem{dong2018time}
X.~Dong, Y.~Li, C.~Lu, G.~Hu, Q.~Li, and Z.~Ren, ``Time-varying formation
  tracking for {UAV} swarm systems with switching directed topologies,''
  \emph{IEEE Trans. Neural Netw. Learn. Syst.}, vol.~30, no.~12, pp.
  3674--3685, 2019.

\bibitem{dong2017time}
X.~Dong and G.~Hu, ``Time-varying formation tracking for linear multiagent
  systems with multiple leaders,'' \emph{IEEE Trans. Autom. Control}, vol.~62,
  no.~7, pp. 3658--3664, 2017.

\bibitem{yu2018practical}
J.~Yu, X.~Dong, Q.~Li, and Z.~Ren, ``Practical time-varying formation tracking
  for second-order nonlinear multiagent systems with multiple leaders using
  adaptive neural networks,'' \emph{IEEE Trans. Neural Netw. Learn. Syst.},
  vol.~29, no.~12, pp. 6015--6025, 2018.

\bibitem{ren2006consensus}
W.~Ren, ``Consensus based formation control strategies for multi-vehicle
  systems,'' in \emph{Proc. Amer. Control Conf.}\hskip 1em plus 0.5em minus
  0.4em\relax IEEE, 2006, pp. 4237--4242.

\bibitem{xiao2009finite}
F.~Xiao, L.~Wang, J.~Chen, and Y.~Gao, ``Finite-time formation control for
  multi-agent systems,'' \emph{Automatica}, vol.~45, no.~11, pp. 2605--2611,
  2009.

\bibitem{yu2011second}
W.~Yu, W.~X. Zheng, G.~Chen, W.~Ren, and J.~Cao, ``Second-order consensus in
  multi-agent dynamical systems with sampled position data,''
  \emph{Automatica}, vol.~47, no.~7, pp. 1496--1503, 2011.

\bibitem{baldi2018output}
S.~Baldi, S.~Yuan, and P.~Frasca, ``Output synchronization of unknown
  heterogeneous agents via distributed model reference adaptation,'' \emph{IEEE
  Trans. Control Netw. Syst.}, vol.~6, no.~2, pp. 515--525, 2018.

\bibitem{yu2012distributed}
W.~Yu, P.~DeLellis, G.~Chen, M.~Di~Bernardo, and J.~Kurths, ``Distributed
  adaptive control of synchronization in complex networks,'' \emph{IEEE Trans.
  Autom. Control}, vol.~57, no.~8, pp. 2153--2158, 2012.

\bibitem{li2013distributed}
Z.~Li, W.~Ren, X.~Liu, and L.~Xie, ``Distributed consensus of linear
  multi-agent systems with adaptive dynamic protocols,'' \emph{Automatica},
  vol.~49, no.~7, pp. 1986--1995, 2013.

\bibitem{cheng2018fully}
B.~Cheng and Z.~Li, ``Fully distributed event-triggered protocols for linear
  multi-agent networks,'' \emph{IEEE Trans. Autom. Control}, vol.~64, no.~4,
  pp. 1655--1662, 2019.

\bibitem{wen2017robust}
G.~Wen, G.~Hu, Z.~Zuo, Y.~Zhao, and J.~Cao, ``Robust containment of uncertain
  linear multi-agent systems under adaptive protocols,'' \emph{Int. J. Robust
  Nonlinear Control}, vol.~27, no.~12, pp. 2053--2069, 2017.

\bibitem{wang2017distributed}
R.~Wang, X.~Dong, Q.~Li, and Z.~Ren, ``Distributed adaptive control for
  time-varying formation of general linear multi-agent systems,'' \emph{Int. J.
  Syst. Sci.}, vol.~48, no.~16, pp. 3491--3503, 2017.

\bibitem{wang2018distributed}
------, ``Distributed time-varying output formation control for general linear
  multiagent systems with directed topology,'' \emph{IEEE Trans. Control Netw.
  Syst.}, vol.~6, no.~2, pp. 609--620, 2018.

\bibitem{yue2020distributed}
D.~{Yue}, J.~{Cao}, Q.~{Li}, and M.~{Abdel-Aty}, ``Distributed neuro-adaptive
  formation control for uncertain multi-agent systems: node- and edge-based
  designs,'' \emph{IEEE Trans. Netw. Sci. Eng.}, 2020.

\bibitem{yu2015distributed}
W.~Yu, J.~Lu, X.~Yu, and G.~Chen, ``Distributed adaptive control for
  synchronization in directed complex networks,'' \emph{SIAM J. Control
  Optim.}, vol.~53, no.~5, pp. 2980--3005, 2015.

\bibitem{yu2018distributed}
Z.~Yu, D.~Huang, H.~Jiang, C.~Hu, and W.~Yu, ``Distributed consensus for
  multiagent systems via directed spanning tree based adaptive control,''
  \emph{SIAM J. Control Optim.}, vol.~56, no.~3, pp. 2189--2217, 2018.

\bibitem{yu2018directed}
Z.~Yu, H.~Jiang, D.~Huang, and C.~Hu, ``Directed spanning tree--based adaptive
  protocols for second-order consensus of multiagent systems,'' \emph{Int. J.
  Robust Nonlinear Control}, vol.~28, no.~6, pp. 2172--2190, 2018.

\bibitem{boyd1994linear}
S.~Boyd, L.~El~Ghaoui, E.~Feron, and V.~Balakrishnan, \emph{Linear Matrix
  Inequalities in System and Control Theory}.\hskip 1em plus 0.5em minus
  0.4em\relax SIAM, 1994, vol.~15.

\bibitem{hu2020distributed}
J.~{Hu}, P.~{Bhowmick}, and A.~{Lanzon}, ``Distributed adaptive time-varying
  group formation tracking for multiagent systems with multiple leaders on
  directed graphs,'' \emph{IEEE Trans. Control Netw. Syst.}, vol.~7, no.~1, pp.
  140--150, 2020.

\bibitem{lv2017novel}
Y.~Lv, Z.~Li, Z.~Duan, and G.~Feng, ``Novel distributed robust adaptive
  consensus protocols for linear multi-agent systems with directed graphs and
  external disturbances,'' \emph{Int. J. Control}, vol.~90, no.~2, pp.
  137--147, 2017.

\end{thebibliography}
\end{document}